\documentclass[a4paper,twocolumn,11pt,unpublished]{quantumarticle}
\pdfoutput=1
\usepackage[utf8]{inputenc}
\usepackage[english]{babel}
\usepackage[T1]{fontenc}
\usepackage{amsmath,amsthm}
\usepackage[shortlabels]{enumitem}
\usepackage{hyperref}
\usepackage{xcolor,soul}
\sethlcolor{lightgray}
\makeatletter
\newcommand{\mylabel}[2]{#2\def\@currentlabel{#2}\label{#1}}
\makeatother
\newtheorem{prop}{Proposition}
\renewenvironment{proof}{{\bf \emph{Proof.} }}{\hfill $\Box$ \\} 
\definecolor{purplee}{rgb}{0.5,0.0,0.5}

\hypersetup{
    colorlinks=true,
    urlcolor=purplee,
    citecolor=purplee,
    linkcolor=purplee
}
\usepackage{amssymb}
\usepackage{tikz}
\usepackage{lipsum}
\usetikzlibrary{external}
\usepackage{booktabs}
\usetikzlibrary{quantikz}
\usepackage{multirow}
\usepackage{longtable}

\newtheorem{theorem}{Theorem}

\begin{document}

\title{Exhaustive Optimisation of Automorphism Groups for Stabiliser Codes}

\author{Aisling Mac Aree}
\email{a.macaree1@universityofgalway.ie}
\orcid{0009000725201289}
\author{Mark Howard}
\email{mark@markhoward.info}
\affiliation{School of Mathematical and Statistical Sciences, University of Galway, Galway, Ireland, H91 TK33}
\orcid{000000026910185X}
\affiliation{School of Mathematical and Statistical Sciences, University of Galway, Galway, Ireland, H91 TK33}

\maketitle

\begin{abstract}
    An important measure of utility for a quantum code is the identification of \textit{which} logical operations can be implemented fault-tolerantly on its codespace. We introduce a framework which leverages the automorphism groups of associated classical codes, the choice of logical basis and exploitation of code equivalence to construct all distinct implementable realisations of each valid logical operation for a given $[[n,k,d]]$ code. We establish conjugacy classes and group transversals (unrelated to transversality) as key explanatory concepts. We subsequently motivate and calculate two figures-of-merit that can be optimised with this framework. Our results yield a table of optimal logical operations and their corresponding physical circuits for all small stabiliser codes with $n \leq 7$ and $k \leq 2$, drawn from quantum databases. This exhaustive table of results provides the optimal physical implementations of logical operations which may be advantageous for both magic state cultivation and experimental purposes.
\end{abstract}

\section{Introduction}\label{sectino intro overall}

Fault-tolerant logical operations are essential for scalable quantum computing. However, the details of which logical operations are available via transversal or SWAP-transversal physical circuits—and how costly these physical operations are—depend on the specific symmetries of the code being used, as well as the chosen cost metric. As the field of quantum error correction has matured past quantum memories, researchers are increasingly concerned with the practicalities of implementing logical gates on the codespace. Motivating much of this research is the fact that, for a given stabiliser code, there are vast families of physical circuits that all implement the same logical unitary. There have been a number of works that are relevant and complementary to the material we outline below. 

Grassl and Rötteler \cite{roetteler} highlighted that the ``SWAP'' part of SWAP-transversal logical gates can lead to reductions in overhead; qubit permutations can be achieved at zero cost via qubit relabelling yet still lead to nontrivial logical gates. Bravyi et al. \cite{Bravyigrosscode} incorporated similar insights using the code symmetries (automorphism groups) of Bivariate Bicycle LDPC codes. A recent paper by Koh et al.~\cite{Koh2026PhantomCodes} extends these ideas, finding that---for quantum codes with a particular feature \textit{(phantom codes)}---all in-block logical entangling gates vanish from the physical circuit after classical circuit compilation and therefore introduce zero overhead and perfect fidelity.
Rengaswamy et al.~\cite{rengaswamy} developed a framework for converting the problem of physical Clifford circuits (implementing a target logical Clifford) into systems of binary equations. Kuehnke et al.~\cite{kuehnke2025hardwaretailoredlogicalcliffordcircuits} further transformed the problem into an integer quadratically constrained program, allowing them to find physical circuits (with e.g., minimal CZ cost) for logical Cliffords. Sayginel et al.~\cite{sayginel} used code and matrix automorphisms to develop a theoretical and software framework for the exhaustive optimisation of any given stabiliser code. Notably, they also rigorously address the Pauli corrections that are lost in the process of representing Clifford unitaries with elements of the binary symplectic group.

This work develops a theoretical approach that captures all possible degrees of freedom, and subsequently cuts down the search space by virtue of structural insights. These mathematical results will likely find application elsewhere. Our work aims to be exhaustive for small stabiliser codes in the sense that we use all small stabiliser codes listed in the QEC database \cite{QECDB}, codetables.de \cite{Codetable} by Markus Grassl, and those listed in the 2025 paper ``Small Binary Stabiliser Subsystem Codes'' by Cross and Vandeth \cite{CrossVandeth2025}. We use these small stabiliser codes as a starting point and, more importantly, we optimise over every degree of freedom that can affect the physical realisation of a logical gate: (i) we optimise over every possible choice of logical $X$ and $Z$ representatives, and (ii) we optimise over every equivalent version of the starting stabiliser code. In the former case, the notion of conjugacy classes helps clarify the relevant mathematical structure and constrain the search space. In the latter case, the notion of a group transversal (distinct from transversal gates) provides the relevant mathematical approach.

We optimise over choice of logical $X$ and $Z$ representatives and equivalent codes based on $2$ independent figures-of-merit. For a given physical circuit $\pi$, the number of SWAP gates and non-identity Clifford operations are denoted by $|\pi_{\sigma}|$ and $|\pi_{e}|$ respectively. We optimise independently over the following two metrics,  
\begin{enumerate}[(i)]
    \item \textbf{Metric 1}, \emph{Control-Clifford Cost}: $7|\pi_{\sigma}| + |\pi_{e}|$
    \item \textbf{Metric 2}, \emph{Local Clifford Cost}: $|\pi_{e}|$.
\end{enumerate}
Broadly speaking, metric $1$ heavily penalises SWAPs whereas metric $2$ counts SWAPs as \textit{free}.

\textit{Controlled-Clifford Cost} (Metric $1$): We minimise a cost associated with measuring a particular logical Clifford given by a stabiliser code. The rationale for doing so stems from magic state distillation (MSD) and cultivation (MSC) \cite{Gidney} protocols. Originating with Knill's $2004$ MSD \cite{Knill2004} scheme and finding recent prominence as the first step in MSC proposals, it is of great interest to find physical circuits that enable non-destructive measurement of logical Clifford observables. For example, within the Steane code, measurement of the Hadamard observable (and hence, preparation of its $\ket{T}$ eigenstate) can be achieved by measuring physical control-Hadamards, by virtue of the fact that logical $H_{L}$ can be implemented by physical $H^{\otimes 7}$. Deleterious aspects of a physical circuit implementation could be:

\begin{enumerate}[i)]
    \item too many local (single-qubit) Clifford gates (since controlling these necessitate non-Clifford gates and hence introduce noise)
    \item SWAPs on physical qubits. When measuring the logical Clifford, these SWAPs become control-SWAPs which necessitate $7$ $T$ gates each (via unitary synthesis, or $4$ non-unitarily \cite{PhysRevA.87.022328, PhysRevLett.118.090501, gidney2018halving}) and moreover, introduce complicated error patterns. For this reason a suitable metric is one that heavily penalises SWAPs (the coefficient $7$ is thus somewhat arbitrary) and only then minimises local Cliffords as much as possible.
\end{enumerate}

\textit{Local Gate Count} (Metric $2$): We minimise the number of physical single-qubit gates while treating qubit permutations (SWAPs) as zero-cost. This metric prioritises the ``zero-cost'' transversal gates identified by Grassl and R\"{o}tteler \cite{roetteler} and has been experimentally validated by Abobeih et al.~\cite{abobeih} (see e.g. Fig. $4$ therein) who demonstrated fault-tolerant operations on a diamond processor by specifically exploiting this permutation-based approach. This low---or no---cost perspective also forms the basis of phantom codes recently developed in \cite{Koh2026PhantomCodes}. We envisage that our database of logical gates with low---or zero---cost physical Cliffords will inform experimental design and find further use in topics like circuit synthesis. 

As application of the theory, our work optimises over the two aforementioned metrics for every automorphism-induced logical operation for all small stabiliser codes (as per \cite{Codetable}, \cite{QECDB}, and \cite{CrossVandeth2025}). Our methodology can be outlined as follows:
\begin{enumerate}
    \item Generate the full automorphism group $\Gamma(\mathcal{C})$ of the classical code $\mathcal{C} \subset \mathrm{GF}(4)^n$ associated with each stabiliser code.
    \item Exploit the freedoms allowed by logical basis choice and code equivalence to generate all distinct physical circuits corresponding to each class of logical transformation.
    \item Optimise each resulting circuit independently to minimise metric $1$ and metric $2$.
\end{enumerate}

Our framework has yielded a results table of optimal implementations, in terms of both of the independent cost metrics, of logical circuits for small stabiliser codes. This table of results can be found in section \ref{section results}.

\subsection{Background of algebraic framework}\label{section intro background}
The Pauli error group on $n$ qubits, denoted $E_{n}$, consists of tensor products of the single-qubit Pauli matrices $\{I,X,Y,Z\}$. The stabiliser group $S \leq E_{n}$\footnote{The stabiliser group $S$ unfortunately shares its notation with the $S$ gate which, in symplectic form, is given by $S= \begin{bmatrix}
    1 & 0 \\
    1 & 1
\end{bmatrix}$. It also bears resemblance to the notation used to describe the symmetric group on $n$ qubits, $S_{n}$. However, in this paper, whether $S$ is referring to the stabiliser group, the $S$ gate or the symmetric group will be clear from context.} is comprised of Pauli operators which leave the code space invariant. We define the stabiliser code $Q(S)$ as the simultaneous $+1$ eigenspace of $S$, which has dimension $2^{k}$ and is generated by a set of $n-k$ mutually commuting independent generators \cite{calderbank1997quantum, Gottesman1997}. The set of logical Pauli operators associated with a given stabiliser code, $Q(S)$, can be identified with the normalizer of the stabiliser group, $N(S)$ \cite{Gottesman1997}. To facilitate algebraic analysis of these codes, it is convenient to quotient out global phases, yielding $\overline{E_{n}}:=E_{n}/\{\pm 1, \pm i\} \cong GF(2)^{2n}$, where elements of this group can be expressed as length $2n$ binary vectors $(x|z)=(x_{1}, \dots, x_{n}|z_{1}, \dots, z_{n}) \longleftrightarrow  (X^{x_{1}} \otimes \dots \otimes X^{x_{n}})(Z^{z_{1}} \otimes \dots \otimes Z^{z_{n}})$. 

Following Calderbank et al.'s well-established formalism \cite{calderbank1997quantum}, we can describe stabiliser codes using additive codes $\mathcal{C} \subset GF(4)^{n}$ where the finite field $GF(4)$ has elements $\{0,1,\omega, \omega^{2}=\bar{\omega} = \omega+1 \}$. This methodology hinges on the existence of a map, $\phi$: \begin{align}\label{eqn phi} \begin{split}
    \phi: GF(4)^{n}& \rightarrow GF(2)^{2n}, \\ \phi(x + \omega z)& = (x|z)
    \end{split}
    \end{align}
which maps classical generators $\mathcal{G}$ to symplectic representations of stabiliser generators $\phi(\mathcal{G})=G$ \footnote{In this paper, we have chosen the convention that objects in $GF(4)^{n}$ are represented calligraphically e.g., $\mathcal{G}, \mathcal{B}, \mathcal{C}$ etc., whereas the image of these objects under $\phi$ which reside in $GF(2)^{2n}$ are represented by plaintext e.g. $G,B,C$ etc.}.

If $\mathcal{C}$ is additive and self-orthogonal, then $\phi(\mathcal{C})=C$ will consist of the binary representation of a commuting set of Pauli operators. Additionally, the corresponding code, $Q(S)$, will have parameters $[[n, n-k, d]]$, where $k=dim(\mathcal{C})$. Thus, given a generator matrix $\mathcal{G} \in GF(4)^{(n-k) \times n}$ for a classical code $\mathcal{C}$, which satisfies the aforementioned conditions, its image under $\phi$, \begin{align}\label{C isom to Sbar} C \cong \overline{S} \leq \overline{E_{n}},\end{align} is generated by $\langle G \rangle$ \cite{calderbank1997quantum}. 

In the symplectic representation, two operators commute if and only if their images under $\phi$ are orthogonal with respect to the \textit{symplectic inner product}. We define the symplectic inner product between $(x|z), (x'|z') \in \mathbb{F}_{2}^{2n}$ as: \begin{align}\langle (x|z), (x'|z')\rangle_{s} = x \cdot z' + z \cdot x' \mod 2.\end{align}

Symplectic matrices, $F \in Sp(2n,2)$, satisfy the symplectic criteria, $F\Omega_{2n} F^{T} = \Omega_{2n}$, such that $\Omega_{2n}= \begin{bmatrix}
    0 & \mathbb{I}_{n} \\
    \mathbb{I}_{n} & 0
\end{bmatrix}$. These $F \in Sp(2n,2)$ preserve the symplectic inner product and map the symplectic representations of Pauli operators to one another.  

The symplectic matrix $F \in Sp(2n,2)$ can be decomposed into a circuit consisting of \textit{Clifford operators}. The single-qubit Clifford group (modulo phases) consists of $24$ elements which can be decomposed into products of the $6$ pure symplectic operators $(I,H,S,HS,SH,HSH)$ and the $4$ single-qubit Pauli gates $(I,X,Y,Z)$\footnote{See section \ref{subsection logical action} for reference to how the partitioning of the Clifford group pertains to \textit{Pauli-correction}.} \cite{AaronsonGottesman, sayginel}. The set of implementable logical Clifford operators on $Q(S)$ coincides with the automorphism group of the associated classical code $\mathcal{C}$. These automorphism-induced logical operators form a subgroup of the symplectic group $Sp(2k,2)$ and are referred to as \textit{SWAP-transversal} gates, meaning they include qubit permutations as well as the traditional weakly transversal implementation of single-qubit Clifford gates \cite{sayginel, Gottesman1997}.  

\section{Logical Operations and Corresponding Physical Circuits}
\subsection{Logical action}\label{subsection logical action}
As noted in Section \ref{section intro background}, the set of logical operators implementable on $Q(S)$ by SWAP-transversal circuits coincides with the automorphism group of the associated classical code $\mathcal{C}$. We denote this automorphism group as $\Gamma(\mathcal{C})$ \cite{calderbank1997quantum, Gottesman1997}. 

We can represent elements of the automorphism group $\Gamma(\mathcal{C})$ as $\pi \in S_{3n}$. As per the Magma algebra system \cite{MR1484478} \textit{``the automorphism group of a length $n$ additive stabiliser code over $F_{4}$ is a subgroup of $Z_{3} \wr S_{3n!}$. However, the automorphism group of the quantum code it generates is a subgroup of $S_{3} \wr S_{n}$ of order $3!n!$ $\dots$ in Magma automorphisms are returned as $\dots$ length $3n$ permutations for the full monomial code''}. Magma output $\pi \in S_{3n}$ can be converted to $(\sigma;\rho_{1}, \dots, \rho_{n}) \in (S_{3})^{n} \rtimes S_{n}$ by splitting the permutation $\pi$ into $n$ blocks of $3$ elements. The ordering of the $n$ blocks gives the global permutation $\sigma$. The ordering of the elements within each of the $n$ blocks corresponds to one of the $6$ single Clifford actions, $\rho_{i}$. Appendix \ref{subsec example converting pi to fpi} gives an explicit example of this conversion from $S_{3n}$ to $(S_{3})^{n} \rtimes S_{n}$. We represent the fact that $(\sigma;\rho_{1}, \dots, \rho_{n}) \in (S_{3})^{n} \rtimes S_{n}$ corresponds to $\pi \in S_{3n}$ by $$\pi \eqsim (\sigma;\rho_{1}, \dots, \rho_{n}).$$

Equation \ref{C isom to Sbar}, combined with the action of the automorphism group $\Gamma(\mathcal{C}): \mathcal{C} \rightarrow \mathcal{C}$, means that the action of $\pi \in \Gamma(\mathcal{C})$ corresponds to that of a Clifford gate. This property allows us to define a physical symplectic matrix, $F_{\pi}$, which maps the vectors corresponding to binary representation of the stabilisers onto themselves, \begin{align}F_{\pi}: C \rightarrow C.\end{align} 
 
Let the mapping from an automorphism to its associated symplectic matrix be described by a homomorphism $\alpha$ \cite{Gottesman1998}, \begin{align}\label{eqn alpha:m pi to Fpi} \begin{split}
     \alpha:\Gamma(\mathcal{C})& \rightarrow Sp(2n,2),\\ \alpha(\pi)& = F_{\pi}.
     \end{split}\end{align}   This mapping is defined by the action of $\pi$ on the vector $x + \omega z \in GF(4)^{n}$, which must correspond to the action of the physical Clifford $F_{\pi}$ on the corresponding Pauli vector $(x|z) \in GF(2)^{2n}$: \begin{align}\label{eqn pi to Fpi^T}\phi(\pi(x + \omega z)) =(x|z)F_{\pi}^{T}.\end{align} The full construction of this homomorphism is provided in Appendix \ref{appx symplectic construction}. 

In order to consider the logical action of an automorphism $\pi \in \Gamma(\mathcal{C})$ on the encoded space, we must first establish a \textit{logical basis} $\mathcal{B} \in \mathbb{F}_{4}^{2k \times n}$, such that $\mathcal{B}_{i} \in \mathcal{C}^{\perp}/\mathcal{C}$ (where $\mathcal{B}_{i}$ denotes row $i$ of $\mathcal{B}$) corresponds to the logical $\mathcal{X}$ or $\mathcal{Z}$ type Pauli operations. This is related to the more common binary form via $\phi(\mathcal{B}) = B$, \footnotesize \begin{align}\label{eqn B arrangement} \mathcal{B} = \begin{bmatrix}
    \mathcal{X}_{L}^{(1)} \\
    \vdots \\
    \mathcal{X}_{L}^{(k)}\\
    \mathcal{Z}_{L}^{(1)}\\
    \vdots \\
    \mathcal{Z}_{L}^{(k)}
\end{bmatrix} \in GF(4)^{2k \times n}, \text{ } B = \begin{bmatrix}
    X_{L}^{(1)} \\
    \vdots \\
    X_{L}^{(k)} \\
    Z_{L}^{(1)} \\
    \vdots \\
    Z_{L}^{(k)}
\end{bmatrix} \in GF(2)^{2k \times 2n}. \end{align} 

\normalsize

A basis $\mathcal{B} \in GF(4)^{2k \times n}$ is considered a valid symplectic basis for a given $[[n,k,d]]$ code if and only if it satisfies:
\begin{align}\label{eqn symplectic basis is symplectic}
\mathcal{B} \odot \mathcal{B}^T=\langle B,B \rangle_s=B \Omega_{2n} B^T=\Omega_{2k},
\end{align}
where the product $\odot$ between any two matrices $\mathcal{U},\mathcal{V} \in GF(4)^{m \times n}$, is described as 
\begin{align*}
    \left[\mathcal{U}\odot  \mathcal{V}^T\right]_{i,j}=Tr(u_i \cdot \overline{v_j}),
\end{align*} where $u_{i}$ represents row $i$ of $\mathcal{U}$ and $\overline{v_{j}}$ denotes the conjugate of row $j$ of $\mathcal{V}$. 
Note that the field-theoretic trace satisfies $Tr(0)=Tr(1)=0$ and $Tr(\omega)=Tr(\omega^2)=1$.

To capture the dependence of the induced logical action on $\pi \in \Gamma(\mathcal{C})$ and $\mathcal{B}$, we describe the logical operation as $L(\pi, \mathcal{B})$. We construct the logical action using the \textit{dual basis} $D_{\mathcal{B}}$ such that \begin{align}\label{eqn destab} \mathcal{D}_{\mathcal{B}} = \Omega_{2k}\mathcal{B} \Rightarrow \mathcal{D}_{\mathcal{B}} \odot \mathcal{B}^T   = \mathbb{I}_{2k}. \end{align}  We obtain $L(\pi, \mathcal{B})$ via the following proposition:
\begin{prop}\label{eqn logical action}
The logical action of the circuit $\pi$ acting on a basis $\mathcal{B}$ is given by the trace inner product of each row of the dual basis $\mathcal{D}_{\mathcal{B}}$ with the transformed basis $\pi(\mathcal{B})^{T}$,
    $$L(\pi, \mathcal{B}) =  \mathcal{D}_{\mathcal{B}} \odot \pi (\mathcal{B})^{T}.$$
\end{prop}
A detailed derivation of Proposition \ref{eqn logical action} can be found in Appendix \ref{appx L(pi,B)}.

In each row of $\mathcal{B}$, we identify $0 \leftrightarrow I$, $1 \leftrightarrow X$, $\omega \leftrightarrow Z$ and $\omega^{2} \leftrightarrow Y$. Consequently, the action $\pi(\mathcal{B})$, refers to the application of $(\rho_{1}, \dots, \rho_{n})$ individually to the row elements corresponding to the $n$ qubits, followed by the application of the permutation $\sigma$ to the columns of $\mathcal{B}$. The mapping between the elements of $GF(4)$ and the phase-free Pauli group is an isomorphism, hence we use their representations in $\mathcal{B}$ interchangeably depending on which makes the most sense in a given situation. 

For example, we take $\mathcal{B}$ to be: $$\mathcal{B} = \begin{bmatrix}
1 & 1 & 0 & 0 \\
\omega^{2} & \omega & 0 & 0 \\
\omega^{2} & 0 & \omega & 0 \\
1 & 0 & 1 & 0
\end{bmatrix} 
$$
Let the circuit $\pi$ be given as:  $$\pi \eqsim ((2,1,3,4); I, I,  HSH, HSH).$$
We begin by applying $HSH$ to column 3 and 4 of $\mathcal{B}$, $$\rho(\mathcal{B}) = \begin{bmatrix}
1 & 1 & 0 & 0 \\
\omega^{2} & \omega & 0 & 0 \\
\omega^{2} & 0 & \omega^{2} & 0 \\
1 & 0 & 1 & 0
\end{bmatrix}$$

We then apply $\sigma= (2,1,3,4)$ to the columns of the resulting $\rho(\mathcal{B})$, 
\small 
$$\sigma(\rho(\mathcal{B})) = \pi(\mathcal{B}) = \begin{bmatrix}
1 & 1 & 0 & 0 \\
\omega & \omega^{2} & 0 & 0 \\
0 & \omega^{2} & \omega^{2} & 0 \\
0 & 1 & 1 & 0 
\end{bmatrix}.$$

We obtain $\mathcal{D}_{\mathcal{B}}$ from $\mathcal{B}$ using Eq. \ref{eqn destab}

$$\mathcal{D}_{\mathcal{B}} = \begin{bmatrix}
    \omega^{2} & 0 & \omega & 0 \\
    1 & 0 & 1 & 0 \\
    1 & 1 & 0 & 0 \\
    \omega^{2} & \omega & 0 & 0 
\end{bmatrix} $$
Using Proposition \ref{eqn logical action}, we obtain the logical operation representing the action of $\pi$ on the basis $\mathcal{B}$, $$L(\pi, \mathcal{B}) = \mathcal{D}_{\mathcal{B}} \odot \pi(\mathcal{B})^{T} = \begin{bmatrix}
    1 & 1 & 1 & 1 \\
    0 & 1 & 1 & 0 \\
    0 & 0 & 1 & 0 \\
    0 & 0 & 1 & 1
\end{bmatrix}.$$ This result, $L(\pi, \mathcal{B})$, corroborates the relevant result given in Section \ref{section results}.

\normalsize

 Symplectic representations encode Clifford operations up to global phases and local Pauli operators. However, the symplectic form factors out phase information and consequently, the physical realisation of $F_{\pi}$ for the associated logical operation $L(\pi, \mathcal{B})$ may differ from the intended logical Clifford by a stabiliser-preserving Pauli phase correction. In practise, this requires applying a suitable Pauli correction to ensure the overall transformation preserves the correct signs of the stabiliser generators. A systematic procedure for determining these Pauli corrections has been developed by Sayginel et al.~\cite{sayginel}, who evaluate the phases acquired by each stabiliser under the proposed Clifford operation and apply the corresponding Pauli adjustments to recover an exact stabiliser mapping. Their framework is fully compatible with the symplectic-Clifford formalism employed in our framework.

\subsection{Conjugacy classes of the symplectic group $Sp(2k,2)$}\label{subsec the symplectic group}
The set of logical operations $L(\pi, \mathcal{B})$ associated with an $[[n,k,d]]$ code form a subgroup of the symplectic group $ Sp(2k,2)$. Elements of $Sp(2k,2)$ are organized into \textit{conjugacy classes}, in which group elements $a$ and $b$ are \textit{conjugate}, $a \sim b$, if they differ by an inner automorphism $b = gag^{-1}$, for $g \in Sp(2k,2)$. The symplectic group grows rapidly with increasing dimensionality; for example, for $k=1$, $Sp(2,2) \cong S_{3}$ has order $6$, while for $k=2$, $Sp(4,2) \cong S_{6}$ has order $720$. As the explicit examples given in this paper is restricted to small stabiliser codes with $k \leq 2$, the corresponding logical groups are exclusively subgroups of the finite groups $Sp(2,2)$ or $Sp(4,2)$ and are therefore of manageable size \footnote{We suspect that more efficient representations of the automorphism group $\Gamma(C)$ and the \textit{transversal group} $\mathcal{T}$ (introduced in section \ref{subsec code equiv}), perhaps combined with further exploitation of code structures, could enable for the computation of larger codes.}. 

Conjugacy in $Sp(2k,2)$ can be shown to be an equivalence relation and as such, the conjugacy classes partition the group. Selecting one element, $A_{j}$, from each class of $Sp(2k,2)$ yields a complete set of class representatives \cite{Dehaene, Taylor}. The automorphism group $\Gamma(\mathcal{C})$ can be grouped based on the conjugacy classes of the logical operations they implement. This grouping has direct operational significance: two logical operations that are conjugate in $Sp(2k,2)$ are equivalent up to a change of logical basis. Circuits implementing conjugate elements therefore realise the same encoded logical operation, differing only by a relabelling of the logical basis. It is worth noting that if a logical operation $L$ is in conjugacy class $[L]$, then its transpose, $L^{T}$, is also a member of $[L]$. Grouping automorphisms by conjugacy class avoids redundant optimisation over circuits that are operationally identical, allowing one representative circuit per class to capture all basis-equivalent implementations.

\newpage 

\onecolumn 

\scriptsize

\begin{tabular}{|p{1.5in}|p{1.2in}|p{1.2in}|p{1.2in}|} \hline 
\multicolumn{4}{|c|}{Conjugacy classes of $Sp(2,2)$} \\
\hline 
\hline
Class Label & Class representative & No. of Elements in Class & Elements in Class\\
\hline
\hline
Conjugacy Class $1$ & $\begin{bmatrix}
    1 & 0 \\
    0 & 1
\end{bmatrix}$ & $1$ & $\{I\}$\\
\hline 
Conjugacy Class $2$ & $\begin{bmatrix}
    0 & 1 \\
    1 & 1
\end{bmatrix}$ & $2$ & \{$SH,HS$\} \\
\hline
Conjugacy Class $3$ & $\begin{bmatrix}
    0 & 1 \\
    1 & 0 
\end{bmatrix}$ & $3$ & $\{H,S,HSH\}$\\
\hline
\end{tabular}\label{table of class reps of SP2,2}

\vspace{1cm}

\begin{tabular}{|p{1.5in}|p{1.5in}|p{2in}|} \hline 
\multicolumn{3}{|c|}{Conjugacy Classes of $Sp(4,2)$} \\
\hline 
\hline
Class Label & Class representative & No. of Elements in Class \\
\hline
\hline
Conjugacy Class $1$ & $\begin{bmatrix}
  1 & 0 & 0 & 0 \\
  0 & 1 & 0 & 0 \\
  0 & 0 & 1 & 0 \\
  0 & 0 & 0 & 1
\end{bmatrix}$ & $1$\\
\hline
Conjugacy Class $2$  & $\begin{bmatrix}
   0 & 0 & 0 & 1 \\
   0 & 0 & 1 & 0 \\
   0 & 1 & 0 & 0 \\
   1 & 0 & 0 & 0 
\end{bmatrix}$ & $15$ \text{ }(e.g. CNOT and SWAP) \\
\hline
Conjugacy Class $3$ & $\begin{bmatrix}
   0 & 0 & 0 & 1 \\
   0 & 0 & 1 & 0 \\
   0 & 1 & 0 & 0 \\
   1 & 0 & 0 & 1
\end{bmatrix}$ & $15$\\
\hline
Conjugacy Class $4$ & $\begin{bmatrix}
   0 & 0 & 0 & 1 \\
   0 & 0 & 1 & 1 \\
   0 & 1 & 1 & 0 \\
   1 & 0 & 0 & 0 
\end{bmatrix}$ & $40$ \\
\hline
Conjugacy Class $5$ & $\begin{bmatrix}
   0 & 0 & 1 & 0 \\
   0 & 0 & 1 & 1 \\
   1 & 1 & 1 & 1 \\
   0 & 1 & 1 & 0 
\end{bmatrix}$ & $40$ \\
\hline
Conjugacy Class $6$  & $\begin{bmatrix}
   0 & 0 & 0 & 1 \\
   0 & 0 & 1 & 1 \\
   1 & 1 & 0 & 0 \\
   1 & 0 & 0 & 0
\end{bmatrix}$ & $45$ \text{ }(e.g. $H_{1} \otimes H_{2}$)\\
\hline
Conjugacy Class $7$ & $\begin{bmatrix}
   0 & 0 & 0 & 1 \\
   0 & 0 & 1 & 0 \\
   0 & 1 & 0 & 0 \\
   1 & 0 & 1 & 0
\end{bmatrix}$ & $90$ \\
\hline
Conjugacy Class $8$ & $\begin{bmatrix}
   0 & 0 & 0 & 1 \\
   0 & 0 & 1 & 1 \\
   1 & 1 & 0 & 1 \\
   1 & 0 & 0 & 0 
\end{bmatrix}$ & $90$ \\
\hline
Conjugacy Class $9$ & $\begin{bmatrix}
   0 & 0 & 0 & 1 \\
   0 & 0 & 1 & 0 \\
   0 & 1 & 0 & 1 \\
   1 & 0 & 1 & 0
\end{bmatrix}$ & $120$ \\
\hline
Conjugacy Class $10$ & $\begin{bmatrix}
   0 & 0 & 0 & 1 \\
   0 & 0 & 1 & 1 \\
   1 & 1 & 0 & 0 \\
   1 & 0 & 1 & 1
\end{bmatrix}$ & $120$ \\
\hline
Conjugacy Class $11$ & $\begin{bmatrix}
   0 & 0 & 0 & 1 \\
   0 & 0 & 1 & 1 \\
   1 & 1 & 0 & 1 \\
   1 & 0 & 1 & 1
\end{bmatrix}$ & $144$ \\
\hline
\end{tabular}\label{table of class reps of SP4,2}

\newpage 

\normalsize

\twocolumn

\subsection{Choice of logical basis}
Exploiting alternative choices of logical basis $\mathcal{B}$ modifies the logical effect of $\Gamma(\mathcal{C})$ on the codespace. $\mathcal{B}$ is a valid logical basis for a code $\mathcal{C}$ with generator $\mathcal{G}$ if and only if it satisfies the following equation: \begin{align}\label{eqn logical basis criteria}
    \begin{bmatrix}
        \mathcal{G} \\ \hline 
        \mathcal{B}
    \end{bmatrix} \odot \begin{bmatrix}
        \mathcal{G} \\ \hline
        \mathcal{B}
    \end{bmatrix}^{T} = \begin{bmatrix}
        0_{(n -k) \times (n-k)} & 0_{(n-k) \times 2k} \\ 0_{2k \times (n-k)} & \Omega_{2k}
    \end{bmatrix}.
\end{align}  A quick calculation shows that $\mathcal{B} \odot \mathcal{B}^{T} = \Omega_{2k} = (A \mathcal{B}) \odot (A \mathcal{B})^{T}$ exactly when $A(\mathcal{B} \odot \mathcal{B}^{T})A^{T} = A \Omega_{2k} A^{T} = \Omega_{2k}$ i.e. when $A$ is a binary symplectic matrix, $A \in Sp(2k,2)$.

A change of logical basis corresponds to a relabelling of the encoded logical operators, while the underlying stabiliser structure and fault-tolerant properties of the code remain invariant. As such, the entire set of logical operations for a given code associated with $\mathcal{C}$ can be constructed by $L(\pi, A\mathcal{B}) = \mathcal{D}_{A\mathcal{B}} \odot \pi(A\mathcal{B})^{T}$, as per Proposition \ref{eqn logical action}.

 \begin{prop}\label{eqn L=ALAinv} Given a logical operation  $L=L(\pi, \mathcal{B})$ (as per Proposition \ref{eqn logical action}) and a symplectic matrix $A \in Sp(2k,2)$, then the resulting Clifford $L(\pi, A\mathcal{B})$ can always be obtained via conjugation of $L$ by $A^{T}$: $$L(\pi, A\mathcal{B}) = (A^{-1})^{T}(L(\pi, \mathcal{B}))A^{T}.$$ \end{prop} 
This result follows directly from the algebraic discussion in subsection \ref{subsec the symplectic group} and its proof can be found in Appendix \ref{appx ALA}. Proposition \ref{eqn L=ALAinv} confirms that basis changes act by conjugation within $Sp(2k,2)$, meaning that all accessible logical operators lie within the same conjugacy class of the subset realised by $\Gamma(C) \leq Sp(2k,2)$. 

All of the foregoing discussion can be summarised in the following theorem: 
\begin{theorem}\label{thm 1}
    Given a $GF(4)$ code $\mathcal{C}$ (corresponding to a stabiliser code $Q(S)$), a code automorphism $\pi$ (corresponding to a physical Clifford circuit that preserves $\mathcal{C}$) and logical basis $\mathcal{B} = [\mathcal{X}_{1}, \dots, \mathcal{Z}_{k}]^{T}$, then the logical Clifford implemented by $\pi$ can be obtained from Proposition \ref{eqn logical action}: $$
    L(\pi, \mathcal{B}) = \mathcal{D}_{\mathcal{B}}  \odot \pi (\mathcal{B})^{T} .
    $$ Furthermore, any other valid choice of logical basis is of the form $A\mathcal{B}$ with $A \in Sp(2k,2)$ and the resulting logical Clifford $L(\pi, A\mathcal{B})$ satisfies Proposition \ref{eqn L=ALAinv}:
    $$L(\pi,A\mathcal{B}) = (A^{-1})^{T}(L(\pi, \mathcal{B}))A^{T}.$$ The notion of conjugacy class therefore fully captures the logical basis degree of freedom.
\end{theorem}

 \begin{proof}
     The full derivations of Proposition \ref{eqn logical action} and \ref{eqn L=ALAinv} are detailed in Appendix \ref{appx L(pi,B)} and \ref{appx ALA} respectively. 
    From proposition \ref{eqn logical action}, we know that the logical action of $\pi \in \Gamma(C)$ acting on a valid basis $\mathcal{B}$ satisfies $$L(\pi, \mathcal{B}) =  \mathcal{D}_{\mathcal{B}} \odot \pi(\mathcal{B})^{T} .$$
     Furthermore, if $L=L(\pi, \mathcal{B})$ is a member of conjugacy class $[L]$ then, by Proposition \ref{eqn L=ALAinv}, we can implement \textit{any} logical operation from that class via a basis change by $A \in Sp(2k,2)$. 
 \end{proof}

We group automorphisms $\pi \in \Gamma(C)$, based on the conjugacy class of the logical action they implement. We can describe the set of automorphisms that implement some logical operation $L(\pi^{i}, \mathcal{B})$ as $[L(\pi^{i}, \mathcal{B})]$. Via basis change by a suitable $A \in Sp(4,2)$, we can implement \textit{any} element of this conjugacy class, with any $\pi^{j}$ that implements a logical operation from $[L(\pi^{i}, \mathcal{B})]$. We give an explicit example of this process in section \ref{section example with fig}.

Recall, for a given circuit, $\pi \in \Gamma(\mathcal{C})$, the number of SWAP gates and non-identity Clifford operations are denoted by the figures-of-merit $|\pi_{\sigma}|$ and $|\pi_{e}|$, respectively. The SWAP count $|\pi_{\sigma}|$ of a circuit $\pi$, refers to the number of SWAP gates, which can be calculated by \begin{align}\label{eqn order of sigma} |\pi_{\sigma}|=n-c, \end{align} where $n$ refers to the number of qubits that $\sigma$ is acting on and $c$ refers to the number of disjoint cycles in $\sigma$'s cycle decomposition. 

We use the figures-of-merit $|\pi_{\sigma}|$ and $|\pi_{e}|$ to optimise independently over the aforementioned metrics,
\begin{enumerate}[i)]
    \item \textbf{Metric $1$:} $7|\pi_{\sigma}| + |\pi_{e}|,$
    \item \textbf{Metric $2$:} $|\pi_{e}|,$
    \end{enumerate}
each of which has differing practical applications as discussed in Section \ref{sectino intro overall}. 

The physical circuits which implement logical actions from a given conjugacy class can vary substantially in resources, such as SWAP count $|\pi_{\sigma}|$ or the number of non-identity Clifford gates $|\pi_{e}|$. If we select a representative within each equivalence class with an optimal cost metric, then we can implement \textit{any} logical action in this equivalence class using these optimised resource costs, by means of a simple basis change.

\subsection{Code equivalence}\label{subsec code equiv}
In addition to the choice of logical basis, code equivalence provides a further degree of freedom for code construction and circuit optimisation. We describe two additive $GF(4)$ codes, $\mathcal{C}$ and $\mathcal{C}'$, as being \textit{equivalent} if there exists a field automorphism and an isometry that maps $\mathcal{C}$ onto $\mathcal{C}'$. The set of all such operations forms the Hamming group $Ham(\mathcal{C})$, also known as the group of monomial automorphisms of additive codes over $GF(4)$. Elements of this group can perform combinations of:
\begin{enumerate}[(i)]
    \item permuting co-ordinates,
    \item multiplying co-ordinates by non-zero field elements,
    \item applying the Frobenius automorphism $x \mapsto x^{2}$ independently to any choice of co-ordinates. 
\end{enumerate} 
Under the mapping $\phi$, the action of (i) corresponds to qubit permutations, whereas the combined action of (ii) and (iii) corresponds to local Clifford operations. As such, elements of the Hamming group, $Ham(\mathcal{C})$, implement SWAP-transversal gates on the codespace associated with $\mathcal{C}$. Structurally, $Ham(\mathcal{C})$ can be realised as $S_{3} \wr_{n} S_{n}$, where $\wr_{n}$ denotes the wreath product on $n$ qubits; that is, the semi-direct product $S_{3}^{n} \rtimes S_{n}$ representing independent local permutations on each co-ordinate, combined with global co-ordinate permutations. This group has order $6^{n}n!$ \cite{calderbank1997quantum, Danielsen}. Equivalently, $Ham(\mathcal{C}) \cong \mathbb{F}_{4}^{*n} \rtimes (\Gamma(\mathbb{F}_{4}) \times S_{n})$.

The automorphism group is necessarily a subgroup of the Hamming group, $\Gamma(\mathcal{C}) \leq Ham(\mathcal{C}) \cong S_{3} \wr_{n} S_{n}$ \cite{calderbank1997quantum, Gottesman1997}; in fact, these automorphisms are exactly the elements of the Hamming group that leave the code unchanged. We can express elements of $\pi \in \Gamma(C)$ as permutations in $S_{3n}$ however, recall that in certain situations, which will be obvious from context, we will also refer to the condensed symplectic form of $\pi$ by the same symbol, but we use ``$\eqsim$'' to flag the abuse of notation in these circumstances.

We can construct a transversal subgroup $\mathcal{T}$ by taking a single representative $\tau$ from each coset of the quotient $Ham(\mathcal{C}) / \Gamma(\mathcal{C})$. The process for constructing $\mathcal{T}$ via computer algebra systems is detailed in appendix \ref{appx: transversal GAP}. If $\begin{bmatrix}
    \mathcal{G} \\ \hline
    \mathcal{B}
\end{bmatrix}$ satisfies the conditions of Eq. \eqref{eqn logical basis criteria} and is associated with a code $\mathcal{C}$, then $\begin{bmatrix}
    \tau(\mathcal{G}) \\ \hline
    \tau(\mathcal{B})
\end{bmatrix} = \begin{bmatrix}
    \mathcal{G}' \\ \hline
    \mathcal{B}'
\end{bmatrix}$ also satisfies the relevant conditions and is associated with a distinct but equivalent code $C'$. Given a code $\mathcal{C}$, the number of distinct equivalent codes is given by \begin{align}\label{eqn order tau}
    |\mathcal{T}|=\frac{6^{n}n!}{|\Gamma(\mathcal{C})|}.
\end{align}
This is essentially the orbit-stabiliser theorem applied to codes. We see an explicit example of Eq.~\eqref{eqn order tau} for the $[[4,2,2]]$ code in Figure~\ref{fig 3 wheels}. Figure~\ref{fig 2 wheels} depicts the size of the automorphism group $|\Gamma(C)|=144$. Fig. \ref{fig 3 wheels} depicts the $216$ distinct equivalent versions of the $[[4,2,2]]$ code. As per Eq.~\eqref{eqn order tau}, $|Ham(\mathcal{C})|=|\mathcal{T}| \times |\Gamma(\mathcal{C})|$, which in the case of the $[[4,2,2]]$ code is $216 \times 144 = 31,104= 6^{4}4! $. As such, taken together, the $144$ automorphisms for each of the $216$ distinct equivalent versions of the $[[4,2,2]]$ code exhaust the Hamming group.

Furthermore, automorphisms of equivalent codes are conjugate and as such the relation \begin{align}\label{eqn pi' = tauinv pi tau}  \pi = \tau \pi' \tau^{-1} \end{align} holds for $\pi \in \Gamma(\mathcal{C})$, the automorphism of the original code, and $\pi'= \tau^{-1}\pi \tau \in \Gamma(\mathcal{C}')$, the corresponding automorphism of the equivalent code $\mathcal{C}'=\tau (\mathcal{C})$. The logical operation describing the effect of the transformed automorphism, $\pi' \in \Gamma(\mathcal{C}')$, on the encoded space of $\mathcal{C}'$ can be obtained via the following equation \begin{align}\label{eqn tau transform}L(\tau, \pi, \mathcal{B}) = \mathcal{D}_{\tau(\mathcal{B})}  \odot ((\tau^{-1}\pi \tau)\tau (\mathcal{B}))^{T}.\end{align} 

\begin{prop}\label{prop code equivalence no effect}
    Code equivalence has no effect on logical operations.
\end{prop}
 Proposition \ref{prop code equivalence no effect} states that if we let $\tau \in \mathcal{T}$, $\pi \in \Gamma(C)$ and  $\tau^{-1}\pi\tau \in \Gamma(\tau(C))$, then the logical operations they implement $L(\pi, \mathcal{B})$ and $L(\tau, \pi, \mathcal{B})$ (constructed according to Proposition \ref{eqn logical action} and Eq.~\eqref{eqn tau transform} respectively) are equal: $$L(\pi, \mathcal{B}) = L(\tau, \pi, \mathcal{B}).$$
 This is proven in Appendix \ref{appendix L(id)=L(tau)}.
Thus, code equivalence alters the physical circuit implementation but preserves the logical action \cite{roetteler}. 

\begin{theorem}\label{thm transversal}
    Let $\mathcal{C}$ be a classical $GF(4)$ code (corresponding to a stabiliser code $Q(S)$) with code automorphism $\pi \in \Gamma(\mathcal{C})$ (corresponding to a physical Clifford circuit that preserves $\mathcal{C}$) and logical basis $\mathcal{B} = [\mathcal{X}_{1}, \dots \mathcal{Z}_{k}]^{T}$ that implements a logical Clifford $L$ as in Theorem \ref{thm 1}. Any and every equivalent stabiliser code $\mathcal{C}'$ is of the form $\mathcal{C}'=\tau(\mathcal{C})$ where $\tau$ is an element of the group transversal $\mathcal{T}$ obtained by the quotient $Ham(\mathcal{C})/\Gamma(\mathcal{C})$. The corresponding code automorphism $\pi'$ for $\mathcal{C}'$ is found via $\pi'=\tau^{-1}\pi\tau$ and its logical action ---on the basis $\tau(\mathcal{B})$---is also $L$. 
\end{theorem}

\begin{proof}
    The proof of theorem \ref{thm transversal} follows from Eq.~\eqref{eqn pi' = tauinv pi tau}, Eq.~\eqref{eqn tau transform} and Proposition \ref{prop code equivalence no effect} (proven in Appendix \ref{appendix L(id)=L(tau)}). 
\end{proof}

This means that exploiting code equivalence does not give us access to new logical operations however, distinct equivalent codes possess distinct automorphism groups, meaning that exploiting code equivalence grants us access to a much larger pool of physical circuits with (potentially) better parameters for practical application.

\section{Worked Example}\label{section example with fig}
 In this section, we use the example of the $[[4,2,2]]$ code to analyse how independent optimisation over the two previously established cost functions (control-Clifford cost and local Clifford cost) improves circuit realisation across distinct symplectic classes. We use the \textit{version} of the $[[4,2,2]]$ code obtained from \href{https://qecdb.org/}{QECDB} \cite{QECDB}, defined by its generator-basis matrix \begin{align}\label{eqn og generator} \begin{bmatrix}
    \mathcal{G} \\
    \hline 
    \mathcal{B}
\end{bmatrix} = \begin{bmatrix}
    1 & 1 & 1 & 1 \\
    \omega & \omega & \omega & \omega \\
    \hline 
    1 & 1 & 0 & 0 \\
    \omega & \omega & 0 & 0 \\
    \omega & 0 & \omega & 0 \\
    1 & 0 & 1 & 0
\end{bmatrix}.
\end{align} The automorphism group of this code, $\Gamma(\mathcal{C})$, contains $144$ distinct physical circuits which, when projected onto $Sp(4,2)$, enact $36$ distinct logical operations, which can be grouped into $6$ non-empty conjugacy classes as shown in Figure~\ref{fig 2 wheels}. Note that $Sp(4,2)$ has $11$ conjugacy classes, as demonstrated by table \ref{table of class reps of SP4,2}, however, in the case of the $[[4,2,2]]$ code, only logical operations from classes $1,2,4,5,6$ and $9$ can be implemented. Furthermore, the labels $1, \dots, 144$ for automorphisms and the index labelling of the conjugacy classes are arbitrary; the only important factor is the grouping of automorphisms according to the logical transformations they induce.

\newpage
\onecolumn

\begin{figure}
  \centering
  \includegraphics[scale=1.53,trim= 119 0 0 0,clip]{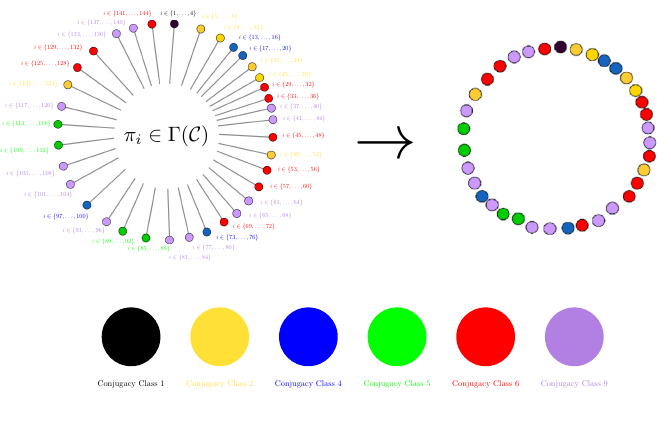}
  \caption{ Figure showing the elements of the automorphism group of $\mathcal{C}$ associated with the $[[4,2,2]]$ code. Each dot represents a distinct logical operation. There are $4$ different automorphisms which implement each logical operation. Elements $\pi^{i} \in \Gamma(\mathcal{C})$ composed in the \textit{same} colour implement logical operations belonging to the \textit{same} conjugacy class. These conjugacy classes are also given in Table \ref{table of class reps of SP4,2}. The first wheel (left) has explicit index labels $\pi^{i} \in \Gamma(\mathcal{C})$, $i \in \{1, \dots, 144\}$, showing the partitioning of the automorphism group based on the logical operations they implement. The second wheel (right) is the same figure shown without the indexing.  }
  \label{fig 2 wheels}
\end{figure}

\begin{figure}
  \centering
  \includegraphics[scale=1.3,trim= 0 65 0 0,clip]{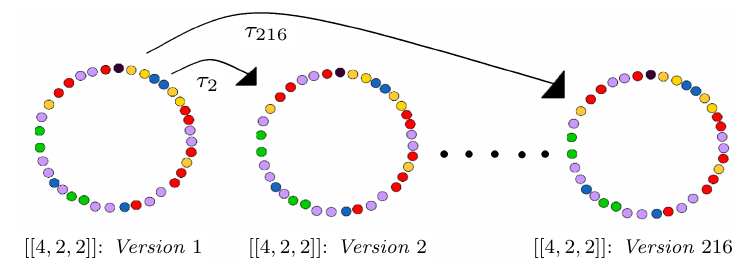}
  \caption{Figure showing the conjugacy classes implemented by the automorphism groups of the $216$ distinct equivalent versions of the $[[4,2,2]]$ code. This figure highlights how the conjugacy classes of the automorphism-induced logical actions are unaffected by code equivalence. Taken together, the $144$ automorphisms for each of the $216$ versions of $[[4,2,2]]$ exhaust the Hamming group of $\mathcal{C}$.}
  \label{fig 3 wheels}
\end{figure}

\newpage

\twocolumn

Every block of four automorphisms in Figure~\ref{fig 2 wheels} implements the \textit{same} logical action, despite the physical difference in the associated circuits. Hence, different automorphisms realising the same logical action may incur different SWAP or Clifford resource costs. 

Consider the automorphism $\pi^{46}$,
printed in red in Fig. \ref{fig 2 wheels} to show that it implements a logical operation from conjugacy class $6$. 
Given that we are working with a code $\mathcal{C}$ defined by the generator-basis matrix given in Eq.~\eqref{eqn og generator} and given that: 
\begin{center} $\pi^{46} \eqsim$
\begin{quantikz}[background color=black!5!white]
\lstick{}  & \push{  \begin{bsmallmatrix}
    1 & 1 \\
    0 & 1
\end{bsmallmatrix}} & \permute{2,3,4,1} & \push{} \\
\lstick{} &  \push{  \begin{bsmallmatrix}
    1 & 1 \\
    0 & 1
\end{bsmallmatrix}} & \push{} & \push{} \\
\lstick{} &  \push{  \begin{bsmallmatrix}
    1 & 1 \\
    0 & 1
\end{bsmallmatrix}}  & \push{} & \push{} \\
\lstick{} &   \push{  \begin{bsmallmatrix}
    1 & 1\\
    0 & 1
\end{bsmallmatrix}} & \push{} & \push{}
\end{quantikz},
\end{center}
we find that $\pi^{46}$ has the following logical effect
\begin{align}
    \mathcal{L}:=L(\pi^{46}, \mathcal{B}) = \begin{bmatrix}
        1 & 1 & 0 & 0 \\
        0 & 1 & 0 & 0 \\
        0 & 1 & 1 & 0 \\
        1 & 1 & 1 & 1
    \end{bmatrix}.
\end{align}  Let us call this logical operation $\mathcal{L}$. 
We can show that $\mathcal{L}$ is a member of Conjugacy Class $6$ by showing that it satisfies Proposition \ref{eqn L=ALAinv}, whereby $\exists A \in Sp(4,2)$ such that $\mathcal{L}$ is conjugate to the conjugacy class $6$ representative.
$$(A^{-1})^{T}\mathcal{L}A^{T} =$$ $$ \begin{bmatrix}
    1 & 0 & 0 & 0 \\
    1 & 0 & 0 & 1 \\
    1 & 1 & 1 & 1 \\
    1 & 1 & 0 & 0 
\end{bmatrix} \begin{bmatrix}
   1 & 1 & 0 & 0 \\
   0 & 1 & 0 & 0 \\
   0 & 1 & 1 & 0 \\
   1 & 1 & 1 & 1
\end{bmatrix} \begin{bmatrix}
   1  & 0 & 0 & 0 \\
   1 & 0 & 0 & 1 \\
   1 & 1 & 1 & 1 \\
   1 & 1 & 0 & 0
\end{bmatrix} =$$ $$ \begin{bmatrix}
    0 & 0 & 1 & 1 \\
    0 & 0 & 1 & 0 \\
    0 & 1 & 0 & 0 \\
    1 & 1 & 0 & 0
\end{bmatrix},$$
which is the class representative of conjugacy class $6$ according to Table \ref{table of class reps of SP4,2}.

As per Eq.~\eqref{eqn order of sigma}, $|\pi^{46}_{\sigma}| = 4-1=3$. This is confirmed by looking at the circuit diagram of $\pi^{46}$ in which the qubits are permuted as follows:  $q_{1} \rightarrow q_{2}, q_{2} \rightarrow q_{3}, q_{3} \rightarrow q_{4}, q_{4} \rightarrow q_{1}$, i.e. $\sigma = (2,3,4,1)$. Therefore, $\sigma$ can be implemented via SWAP$_{1,2}$SWAP$_{2,3}$SWAP$_{3,4}$ as demonstrated by the following circuit diagram: 

\begin{center}
   $\pi^{46}_{\sigma} \eqsim$ \begin{quantikz}[background color=black!5!white]
        \lstick{} & \permute{1,2,4,3} & \permute{1,3,2,4} & \permute{2,1,3,4} \\
        \lstick{} & & &  \\
        \lstick{} & & &  \\
        \lstick{} & & &
    \end{quantikz}.
\end{center}

Therefore, $\mathcal{L}$ can be implemented with the resources $|\pi^{46}_{\sigma}|=3$ and---as can be read from the circuit diagram---$|\pi^{46}_{e}|=4$, meaning that it has a control-Clifford cost of $7(3)+4=25$ and a local Clifford cost of value $|\pi_{e}|=4$. However, both of these costs can be improved upon by exploiting basis changes and code equivalence. In this example, we will optimise over control-Clifford cost (metric $1$) only.

 Searching among the other $36$ automorphisms (see Figure~\ref{fig 2 wheels}) which implement logical operations in Class $6$, $[\mathcal{L}]$ (i.e. every automorphism indexed in red in Figure~\ref{fig 2 wheels}), we find that  
\begin{center} $\pi^{30} \eqsim$
\begin{quantikz}[background color=black!5!white]
\lstick{} &   \push{  \begin{bsmallmatrix}
    1 & 1 \\
    0 & 1
\end{bsmallmatrix}} & \permute{2,1,3,4} & \push{} \\
\lstick{} &  \push{  \begin{bsmallmatrix}
    1 & 1 \\
    0 & 1
\end{bsmallmatrix}} & \push{} & \push{} \\
\lstick{} & \push{  \begin{bsmallmatrix}
    1 & 1 \\
    0 & 1
\end{bsmallmatrix}}  & \push{} & \push{} \\
\lstick{}  & \push{  \begin{bsmallmatrix}
    1 & 1\\
    0 & 1
\end{bsmallmatrix}} & \push{} & \push{}
\end{quantikz}
\end{center}
has parameters $|\pi^{30}_{\sigma}|=1$ and $|\pi^{30}_{e}|=4$, meaning that $\pi^{30}$ has control-Clifford cost $7(1) + 4 = 11$. This is a significant improvement the control-Clifford cost of $\pi^{46}$. Despite, $\pi^{30}$ not implementing $\mathcal{L}$, it does implement a logical operation in the same conjugacy class as our desired logical operation $$L(\pi^{30}, \mathcal{B}) = \begin{bmatrix}
     1 & 1 & 1 & 1 \\
     0 & 1 & 1 & 0 \\
     0 & 0 & 1 & 0 \\
     0 & 0 & 1 & 1
 \end{bmatrix} \in [\mathcal{L}].$$ By Proposition \ref{eqn L=ALAinv}, we are guaranteed that $\exists \text{ } A \in Sp(4,2)$ such that $ L(\pi^{30},A\mathcal{B})=\mathcal{L}$, i.e. $\exists$ $A \in Sp(4,2)$ which satisfies $\mathcal{L} = (A^{-1})^{T} (L(\pi^{30}, \mathcal{B}))A^{T}.$ One such $A$ is:
\begin{align}\label{Amat}A = \begin{bmatrix}
 0 & 0 & 0 & 1 \\
 0 & 0 & 1 & 0 \\
 0 & 1 & 0 & 0 \\
 1 & 0 & 0 & 0
\end{bmatrix},\end{align}
meaning that we would define the generator-basis matrix as:
$$\begin{bmatrix}
    \mathcal{G} \\ \hline A\mathcal{B} 
\end{bmatrix} = \begin{bmatrix}
    1 & 1 & 1 & 1 \\
    \omega & \omega & \omega & \omega \\
    \hline 
1 & 0 & 1 & 0 \\
\omega & 0 & \omega & 0 \\
\omega & \omega & 0 & 0 \\
1 & 1 & 0 & 0 
\end{bmatrix}.$$

We can improve the control-Clifford cost of implementing $\mathcal{L}$ even further by exploiting code equivalence. By Eq.~\eqref{eqn order tau}, the $[[4,2,2]]$ code admits $216$ distinct equivalent representations, each with a unique automorphism group. These equivalent representations of codes are available under transformation by elements of the transversal group $\tau \in \mathcal{T}$ (Figure~\ref{fig 2 dots} demonstrates this concept and Appendix \ref{appx: transversal GAP} details how to construct $\mathcal{T}$ via GAP). Theorem \ref{thm transversal} shows that code equivalence preserves logical action, but not physical circuit implementation i.e. $L(\pi, \mathcal{B}) = L(\pi, \tau(\mathcal{B})), \text{ } \Gamma(\tau(\mathcal{C})) \neq \Gamma(\mathcal{C})$.

We find that optimal implementation of $\mathcal{L}$ in terms of Clifford count can be obtained by use of the transversal element $\tau$,
\begin{center}\label{tau diagram} $\tau \eqsim$
\begin{quantikz}[background color=black!5!white]
\lstick{} & \push{  \begin{bsmallmatrix}
    1 & 1 \\
    0 & 1
\end{bsmallmatrix}} &\push{} & \\
\lstick{}  & \push{} & \push{} & \\
\lstick{}  & \push{}  & \push{} & \\
\lstick{}   & \push{} & \push{} &
\end{quantikz}.
\end{center}
We transform the generator-basis matrix to obtain the equivalent code $\tau(\mathcal{C})= \mathcal{C'}$
$$\begin{bmatrix}
    \tau(\mathcal{G}) \\
    \hline \tau(\mathcal{B})
\end{bmatrix} = \begin{bmatrix}
   1 & 1 & 1 & 1 \\
   \omega^{2} & \omega & \omega & \omega \\
    \hline 
    1 & 1 & 0 & 0 \\
    \omega^{2} & \omega & 0 & 0 \\
    \omega^{2} & 0 & \omega & 0 \\
    1 & 0 & 1 & 0
\end{bmatrix}.$$
Transformation of $\pi^{30}$ by $\tau$ yields $\pi'^{30} \in \Gamma(\mathcal{C}')$:
\smallskip
\begin{center} $\pi'^{30} \eqsim$
\begin{quantikz}[background color=black!5!white]
\lstick{} & \push{} & \permute{2,1,3,4} & \push{} &\\
\lstick{} & \push{} & \push{} & \push{} & \\
\lstick{} & \push{  \begin{bsmallmatrix}
    1 & 1 \\
    0 & 1
\end{bsmallmatrix}} & \push{}  & \push{} &\\
\lstick{} &  \push{  \begin{bsmallmatrix}
    1 & 1\\
    0 & 1
\end{bsmallmatrix}} & \push{} & \push{} &
\end{quantikz},
\end{center}
which has the ideal cost of $7|\pi'^{30}_{\sigma}| + |\pi'^{30}_{e}|=9$ and implements a logical operation in $[\mathcal{L}]$: $$L(\pi'^{30},\tau(\mathcal{B}))= \begin{bmatrix}
    1 & 1 & 1 & 1 \\
    0 & 1 & 1 & 0 \\
    0 & 0 & 1 & 0 \\
    0 & 0 & 1 & 1
\end{bmatrix} \in [\mathcal{L}].$$ Once again, according to Proposition \ref{eqn L=ALAinv}, $\exists \text{ } A \in Sp(4,2)$ such that \begin{equation}\label{pi46 change of basis}
    (A^{-1})^{T}(L(\pi'^{30},\tau(\mathcal{B}))) A^{T}=\mathcal{L},
\end{equation}
We use the $A$ given in Eq.~\eqref{Amat} because (as per Theorem \ref{thm 1}) $L(\pi^{30},\mathcal{B}) =L(\pi'^{30},\tau(\mathcal{B}))$. As such, we can define the code of interest by its generator-basis matrix

\begin{align}\label{GBtauA}\begin{bmatrix}
    \mathcal{\tau(G)} \\ \hline \tau(A\mathcal{B})
\end{bmatrix} = \begin{bmatrix}
1 & 1 & 1 & 1 \\
\omega^{2} & \omega & \omega & \omega 
\\ \hline
    1 & 0 & 1 & 0 \\
    \omega^{2} & 0 & \omega & 0 \\
    \omega^{2} & \omega & 0 & 0 \\
    1 & 1 & 0 & 0
\end{bmatrix}.\end{align}

Finally, according to our framework, $$L(\pi'_{30}, \tau(A\mathcal{B})) = \begin{bmatrix}
   1 & 1 & 0 & 0 \\
   0 & 1 & 0 & 0 \\
   0 & 1 & 1 & 0 \\
   1 & 1 & 1 & 1
\end{bmatrix}= \mathcal{L}.$$ 
Therefore, $\mathcal{L}$ can be implemented with the physical circuit 
$\pi'_{30}$ with control-Clifford cost $7|\pi^{30}_{\sigma}|+|\pi'^{30}_{e}|=9$. This proves that code equivalence not only significantly enlarges the search space of valid physical circuits, but also that this broader circuit pool can produce strictly better implementations than those obtainable from basis change alone.

\vspace{0.5cm}

\onecolumn

\begin{tabular}{|p{1.6in}|p{1.6in}|p{1.6in}|p{1.2in}|}\hline 
\multicolumn{4}{|c|}{Summary of $[[4,2,2]]$ example for controlled-Clifford metric $7|\pi^{i}_{\sigma}| + |\pi^{i}_{e}|$} \\ \hline 
\textbf{Generator-Basis Matrix}  & \textbf{Automorphism $\pi^{i}$} & \textbf{Logical Operation} & \textbf{Metric cost}  \\
\hline
$\begin{bmatrix}
    \mathcal{G} \\ \hline \mathcal{B}
\end{bmatrix}$  & $\pi^{46}$ & $\mathcal{L}$& $25$  \\
 \hline

 $\begin{bmatrix}
    \mathcal{G} \\ \hline \mathcal{B}
\end{bmatrix}$  & $\pi^{30}$ & $[\mathcal{L}]\ni L(\pi^{30},\mathcal{B}) \neq \mathcal{L}$ & $11$  \\
 \hline
 
 $\begin{bmatrix}
    \mathcal{G} \\ \hline A\mathcal{B}
\end{bmatrix}$  & $\pi^{30}$ & $\mathcal{L}$ & $11$   \\
\hline 

  $\begin{bmatrix}
    \tau(\mathcal{G}) \\ \hline \tau(A\mathcal{B})
\end{bmatrix}$  & $\pi'^{30}$ & $\mathcal{L}$ & $9$  \\
 \hline
\end{tabular}\label{table of results for example with figure}

\begin{figure}
  \centering
  \includegraphics[scale=0.8,trim= 10 120 120 20,clip]{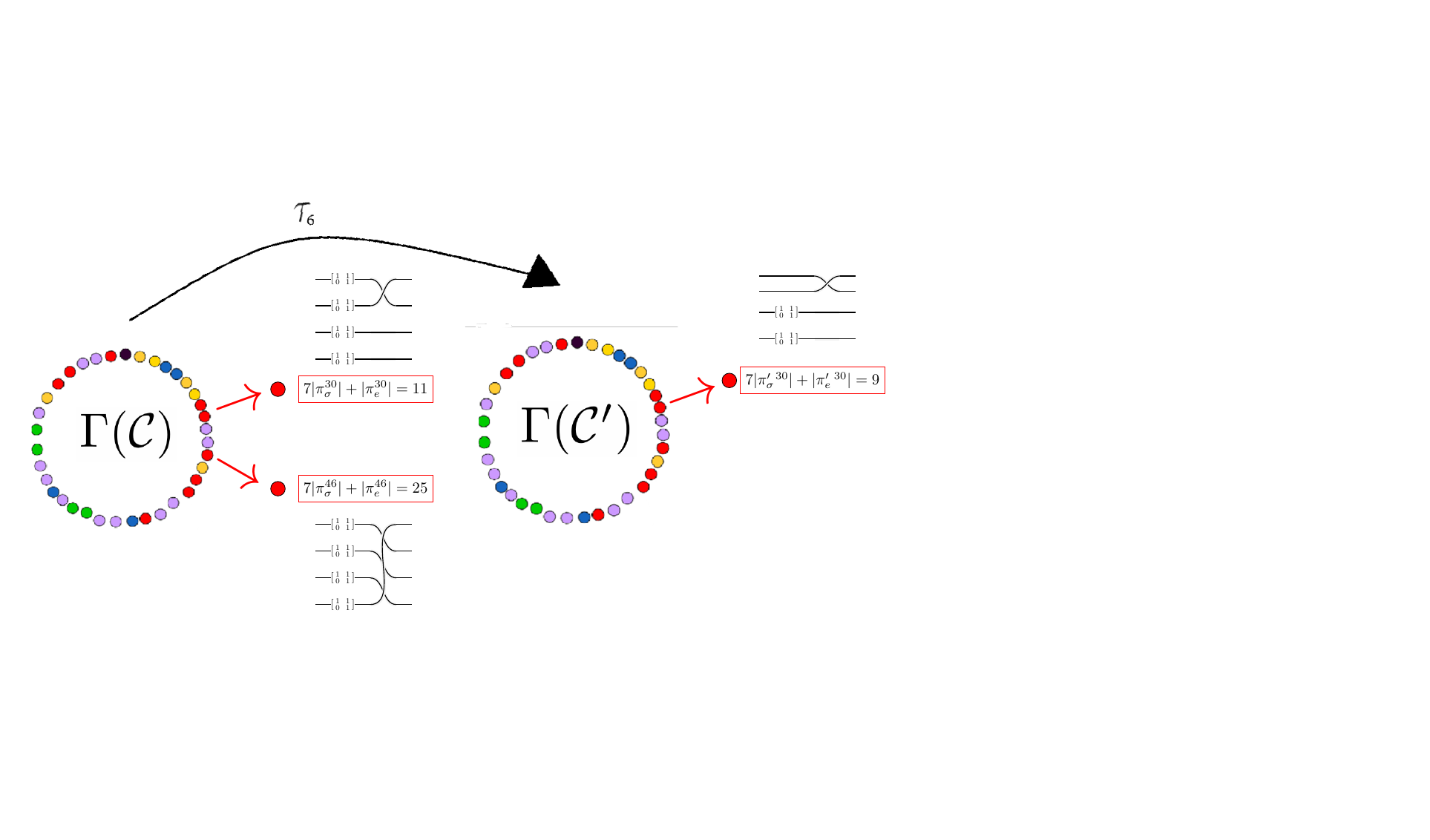}
  \caption{ Relating to $[[4,2,2]]$ example, showing that the best implementation of $\mathcal{L}$ is given by $\pi'^{30}$. $\Gamma(\mathcal{C'})$ denotes the automorphism group obtained by using (in our case), the $6^{th}$ element of the transversal group $\tau^{6} \in \mathcal{T}$ (the circuit diagram of this transversal $\tau$ is given explicitly in the example).  }
  \label{fig 2 dots}
\end{figure}

\newpage

\onecolumn
\section{Results}\label{section results}
The same methodology as detailed in the example of Sec.~\ref{section example with fig} was applied uniformly to all small stabiliser codes; the corresponding results are summarised in the following tables.

\footnotesize
\newsavebox{\circA}
\newsavebox{\circB}
\newsavebox{\circC}
\newsavebox{\circD}
\newsavebox{\circE}
\newsavebox{\circF}
\newsavebox{\circG}

\savebox{\circA}{%
\begin{quantikz}
\lstick{} & \push{} & \permute{3,2,4,1} & \push{} & \\
\lstick{} & \push{  \begin{bsmallmatrix}
    0 & 1 \\
    1 & 1
\end{bsmallmatrix}} &  & \push{} & \\
\lstick{} & \push{} &  \push{} & \push{} \\
\lstick{} & \push{} &  \push{} & \push{}
\end{quantikz}%
}

\savebox{\circB}{%
\begin{quantikz}
\lstick{} & \push{\begin{bsmallmatrix}
    0 & 1 \\
    1 & 0 
\end{bsmallmatrix}} & \permute{1,4,3,2}  & \push{} & \push{}  \\
\lstick{} & \push{} & \push{}  & \push{} & \push{} & \\
\lstick{} & \push{\begin{bsmallmatrix}
    0 & 1 \\
    1 & 0 
\end{bsmallmatrix}}  & \push{}  & \push{}  & \push{}\\
\lstick{} & \push{}  & \push{} & \push{} & \push{}
\end{quantikz}%
}

\savebox{\circC}{%
\begin{quantikz}
\lstick{} &  & \push{  \begin{bsmallmatrix}
    1 & 1 \\
    0 & 1
\end{bsmallmatrix}} &\push{} & \push{} \\
\lstick{} & & \push{  \begin{bsmallmatrix}
    1 & 1 \\
    0 & 1
\end{bsmallmatrix}} & \push{} & \push{} \\
\lstick{} & & \push{  \begin{bsmallmatrix}
    1 & 1 \\
    0 & 1
\end{bsmallmatrix}}  & \push{} & \push{} \\
\lstick{} &  & \push{  \begin{bsmallmatrix}
    1 & 1\\
    0 & 1
\end{bsmallmatrix}} & \push{} & \push{}
\end{quantikz}%
}

\savebox{\circD}{%
\begin{quantikz}
\lstick{} &  & \push{  \begin{bsmallmatrix}
    0 & 1 \\
    1 & 1
\end{bsmallmatrix}} &\push{} & \push{}\\
\lstick{} & & \push{  \begin{bsmallmatrix}
    0 & 1 \\
    1 & 1
\end{bsmallmatrix}} & \push{} & \push{} \\
\lstick{} & & \push{  \begin{bsmallmatrix}
    0 & 1 \\
    1 & 1
\end{bsmallmatrix}}  & \push{} & \push{} \\
\lstick{} &  & \push{  \begin{bsmallmatrix}
    0 & 1\\
    1 & 1
\end{bsmallmatrix}} & \push{} & \push{}
\end{quantikz}%
}

\savebox{\circE}{%
\begin{quantikz}
\lstick{} & \push{} & \permute{2,3,1,4}  & \push{}  &\\
\lstick{} & \push{} & \push{} & \push{}  \\
\lstick{} & \push{}  & \push{}  & \push{} \\
\lstick{} & \push{  \begin{bsmallmatrix}
    0 & 1\\
    1 & 1
\end{bsmallmatrix}} & \push{} & \push{} &
\end{quantikz}%
}

\savebox{\circF}{%
 \begin{quantikz}
\lstick{} & \push{} & \permute{2,1,3,4} & \push{}  &\\
\lstick{} & \push{} & \push{} & \push{} & \\
\lstick{} & \push{  \begin{bsmallmatrix}
    1 & 1 \\
    0 & 1
\end{bsmallmatrix}} & \push{}  & \push{} &\\
\lstick{} & \push{  \begin{bsmallmatrix}
    1 & 1\\
    0 & 1
\end{bsmallmatrix}} & \push{} & \push{} &
\end{quantikz}%
}

\savebox{\circG}{%
\begin{quantikz}
\lstick{} & \push{} & \permute{2,1,3,4} & \push{} &\\
\lstick{} & \push{  \begin{bsmallmatrix}
    1 & 1 \\
    1 & 0
\end{bsmallmatrix}} & \push{} & \push{} & \\
\lstick{} &  \push{  \begin{bsmallmatrix}
    0 & 1 \\
    1 & 1
\end{bsmallmatrix}}  & \push{} & \push{} &\\
\lstick{} & \push{  \begin{bsmallmatrix}
    0 & 1\\
    1 & 1
\end{bsmallmatrix}} & \push{} & \push{} &
\end{quantikz}%
}

\begin{tabular}{|c|c|c|c|c|c|}  
 \hline 
 \multicolumn{6}{|c|}{Optimising over controlled-Clifford metric $7|\pi^{ }_{\sigma}| + |\pi^{ }_{e}|$}\\
 \hline
Code & Class & Physical Circuit $\pi$ & $7|\pi_{\sigma}| + |\pi_{e}|$ & 
$\left[\begin{array}{cccc}
 \mathcal{G} \\ \hline \mathcal{B}
 \end{array}\right]$
&  $L(\pi, \mathcal{B})$  \\
\hline
$[[4,1,2]]$ & $2$ &  
\usebox{\circA} 

 & $15$ & $\left[\begin{array}{cccc}
    1 & \omega^{2} & \omega^{2} & 0 \\
    0 & 1 & 1 & \omega^{2} \\
    \omega & 0 & \omega & \omega \\ \hline 
    \omega & 1 & 0 & 0 \\
    0 & \omega & \omega & 0 
\end{array}\right]$ & $\begin{bmatrix}
    0 & 1 \\
    1 & 1
\end{bmatrix}$ \\ \hline

& $3$ &  \usebox{\circB}

& $9$ & $\left[\begin{array}{cccc}
    1 & \omega^{2} & \omega & 0 \\
    0 & 1 & \omega^{2} & 1 \\
    \omega & 0 & 1 & \omega^{2} \\ \hline 
    \omega & 1 & 0 & 0 \\
    0 & \omega & 1 & 0 
\end{array}\right]$ & $\begin{bmatrix}
    0 & 1 \\
    1 & 0 
\end{bmatrix}$ \\ \hline
$[[4,2,2]]$ & $2$ &
\usebox{\circC}
  & $4$ & $\left[\begin{array}{cccc}
    1 & 1 & 1 & 1 \\
    \omega & \omega & \omega & \omega \\ \hline
    1 & 1 & 0 & 0 \\
    \omega & \omega & 0 & 0 \\
    \omega & 0 & \omega & 0 \\
    1 & 0 & 1 & 0 
\end{array}\right]$ & $\begin{bmatrix}
  1 & 1 & 0 & 0 \\
  0 & 1 & 0 & 0 \\
  0 & 0 & 1 & 0 \\
  0 & 0 & 1 & 1
\end{bmatrix}$ \\
  \hline 

  & $4$ & \usebox{\circD} & $4$ & $\left[\begin{array}{cccc}
    1 & 1 & 1 & 1 \\
    \omega & \omega & \omega & \omega \\
    \hline 
    1 & 1 & 0 & 0 \\
    \omega & \omega & 0 & 0 \\
    \omega & 0 & \omega & 0 \\
    1 & 0 & 1 & 0 
\end{array}\right]$ & $\begin{bmatrix}
    0 & 1 & 0 & 0 \\
    1 & 1 & 0 & 0 \\
    0 & 0 & 1 & 1 \\
    0 & 0 & 1 & 0 
\end{bmatrix}$ \\
\hline 
& $5$ &  \usebox{\circE}
& 15 & $\left[\begin{array}{cccc}
    \omega^{2} & \omega & 1 & 1 \\
    1 & \omega^{2} & \omega & \omega \\
    \hline 
    \omega^{2} & \omega & 0 & 0 \\
    1 & \omega^{2} & 0 & 0 \\
    1 & 0 & \omega & 0 \\
    \omega^{2} & 0 & 1 & 0 
\end{array}\right]$ & $\begin{bmatrix}
    0 & 1 & 1 & 0 \\
    1 & 1 & 1 & 1 \\
    1 & 1 & 0 & 0 \\
    0 & 1 & 0 & 0
\end{bmatrix}$ \\ \hline 
& $6$ & \usebox{\circF} & $9$ & $\left[\begin{array}{cccc}
    1 & 1 & 1 & 1 \\
    \omega^{2} & \omega & \omega & \omega \\ \hline 
    1 & 1 & 0 & 0 \\
    \omega^{2} & \omega & 0 & 0 \\
    \omega^{2} & 0 & \omega & 0 \\
    1 & 0 & 1 & 0 
\end{array}\right]$ & $\begin{bmatrix}
1 & 1 & 1 & 1 \\
0 & 1 & 1 & 0 \\
0 & 0 & 1 & 0 \\
0 & 0 & 1 & 1
\end{bmatrix}$ \\
\hline & $9$ & \usebox{\circG} & $10$ & $\left[\begin{array}{cccc}
    \omega & 1 & 1 & 1 \\
    \omega^{2} & \omega & \omega & \omega \\
    \hline 
    \omega & 1 & 0 & 0 \\
    \omega^{2} & \omega & 0 & 0 \\
    \omega^{2} & 0 & \omega & 0 \\
    \omega & 0 & 1 & 0
\end{array}\right]$ & $\begin{bmatrix}
   0 & 1 & 1 & 0 \\
   1 & 1 & 1 & 1 \\
   0 & 0 & 1 & 1 \\
   0 & 0 & 1 & 0 
\end{bmatrix}$ \\
\hline
\end{tabular}

\newsavebox{\circAA}
\newsavebox{\circBB}
\newsavebox{\circCC}
\newsavebox{\circDD}
\newsavebox{\circEE}
\newsavebox{\circFF}
\newsavebox{\circGG}

\savebox{\circAA}{%
\begin{quantikz}
\lstick{} & \push{  \begin{bsmallmatrix}
    0 & 1 \\
    1 & 0
\end{bsmallmatrix}} & \permute{1,5,3,4,2}  &\push{} &\\
\lstick{} & \push{} & \push{} & \push{} & \\
\lstick{} & \push{  \begin{bsmallmatrix}
    0 & 1 \\
    1 & 0
\end{bsmallmatrix}} &   & \push{} &\\
\lstick{} & \push{  \begin{bsmallmatrix}
    0 & 1\\
    1 & 0
\end{bsmallmatrix}} &  & \push{} & \\
\lstick{} & \push{} & \push{} & \push{} &
\end{quantikz}%
}

\savebox{\circBB}{%
\begin{quantikz}
\lstick{} &  & \push{  \begin{bsmallmatrix}
    1 & 1 \\
    1 & 0
\end{bsmallmatrix}} &\push{} & \push{}\\
\lstick{} & & \push{  \begin{bsmallmatrix}
    1 & 1 \\
    1 & 0
\end{bsmallmatrix}} & \push{} & \push{} \\
\lstick{} & & \push{  \begin{bsmallmatrix}
    0 & 1 \\
    1 & 1
\end{bsmallmatrix}}  & \push{} & \push{} \\
\lstick{} &  & \push{  \begin{bsmallmatrix}
    1 & 1\\
    1 & 0
\end{bsmallmatrix}} & \push{} & \push{} \\
\lstick{} &  & \push{  \begin{bsmallmatrix}
    0 & 1\\
    1 & 1
\end{bsmallmatrix}} & \push{} & \push{}
\end{quantikz}%
}

\savebox{\circCC}{%
\begin{quantikz}
\lstick{} &  \push{  \begin{bsmallmatrix}
    0 & 1 \\
    1 & 0
\end{bsmallmatrix}} & \permute{1,2,3,5,4} &\push{} &\\
\lstick{} & \push{  \begin{bsmallmatrix}
    1 & 0 \\
    1 & 1
\end{bsmallmatrix}} & \push{} & \push{} & \\
\lstick{} &  \push{  \begin{bsmallmatrix}
    0 & 1 \\
    1 & 0
\end{bsmallmatrix}} & \push{} & \push{} &\\
\lstick{} & \push{} & \push{} & \push{} & \\
\lstick{} & \push{} & \push{} & \push{} &
\end{quantikz}%
}

\savebox{\circDD}{%
\begin{quantikz}
\lstick{} & \push{}  & \push{} & \permute{2,1,5,4,3} & \push{} \\
\lstick{} & \push{} & \push{} & \push{} & \push{} \\
\lstick{} & \push{} & \push{} & \push{} & \push{}\\
\lstick{} & \push{} & \push{} & \push{} & \push{} \\
\lstick{} & \push{} & \push{} & \push{} & \push{}
\end{quantikz}%
}

\savebox{\circEE}{%
\begin{quantikz}
\lstick{}  & &  \push{} & \push{} & \push{}\\
\lstick{} & & \push{  \begin{bsmallmatrix}
    1 & 1 \\
    0 & 1
\end{bsmallmatrix}} & \push{} & \push{}  \\
\lstick{} & &  \push{} & \push{} & \push{}\\
\lstick{} & &  \push{} & \push{} & \push{}\\
\lstick{} & &  \push{} & \push{} & \push{}
\end{quantikz}%
}

\savebox{\circFF}{%
\begin{quantikz}
\lstick{} & & \push{  \begin{bsmallmatrix}
    1 & 1 \\
    0 & 1
\end{bsmallmatrix}} & \push{} & \push{} \\
\lstick{} & & \push{  \begin{bsmallmatrix}
    1 & 1 \\
    0 & 1
\end{bsmallmatrix}} & \push{} & \push{} \\
\lstick{} & &  \push{} & \push{} & \push{} \\
\lstick{} & &  \push{} & \push{} & \push{} \\
\lstick{} & &  \push{} & \push{} & \push{}
\end{quantikz}%
}

\savebox{\circGG}{%
\begin{quantikz}
\lstick{} &  & \push{} & \permute{2,1,5,4,3} & \\
\lstick{} & & \push{  \begin{bsmallmatrix}
    1 & 1 \\
    0 & 1
\end{bsmallmatrix}} & \push{} &  \\
\lstick{} & &  \push{} & \push{} & \\
\lstick{} & &  \push{} & \push{} & \\
\lstick{} & &  \push{} & \push{} &
\end{quantikz}%
}

\begin{tabular}{|c|c|c|c|c|c|}  
 \hline 
 \multicolumn{6}{|c|}{Optimising over controlled-Clifford metric $7|\pi^{ }_{\sigma}| + |\pi^{ }_{e}|$}\\

 \hline
Code & Class & Physical Circuit $\pi$ & $7|\pi_{\sigma}| + |\pi_{e}|$ & $\begin{bmatrix}
    \mathcal{G} \\ \hline \mathcal{B}
\end{bmatrix}$ &  $L(\pi, \mathcal{B})$  \\

\hline

$[[5,1,2]]$ & $3$ & \usebox{\circAA} & $10$ & $\begin{bmatrix}
    \omega & \omega & 0 & \omega & 0 \\
    0 & 0 & \omega & \omega & 1 \\
    0 & 1 & 1 & 1 & 0 \\
    1 & 0 & 0 & 1 & \omega \\ \hline 
    0 & \omega & 0 & \omega & 1 \\
    0 & 0 & 1 & 0 & \omega
\end{bmatrix}$ & $\begin{bmatrix}
    0 & 1 \\
    1 & 0 
\end{bmatrix}$ \\
\hline 
$[[5,1,3]]$ & $2$ & \usebox{\circBB} & $5$ & $\begin{bmatrix}
\omega^{2} & 1 & \omega & 1 & 0 \\
1 & 0 & \omega & \omega^{2} & 1 \\
0 & 1 & 1 & \omega^{2} & \omega \\
1 & \omega^{2} & 0 & 1 & \omega \\ \hline 
\omega & 1 & 1 & 0 & 0 \\
1 & \omega^{2} & \omega & 0 & 0
\end{bmatrix}$ & $\begin{bmatrix}
    0 & 1 \\
    1 & 1
\end{bmatrix}$ \\
\hline & $3$ & \usebox{\circCC} & $10$ & $\begin{bmatrix}
\omega^{2} & 1 & \omega & 1 & 0 \\
1 & 0 & \omega & \omega^{2} & 1 \\
0 & 1 & 1 & \omega^{2} & \omega \\
1 & \omega^{2} & 0 & 1 & \omega \\ \hline 
\omega & 1 & 1 & 0 & 0 \\
1 & \omega^{2} & \omega & 0 & 0
\end{bmatrix}$ & $\begin{bmatrix}
    0 & 1 \\
    1 & 0
\end{bmatrix}$ \\ \hline

$[[5,2,1]]$ & $2$ & \usebox{\circDD} & $14$ & $\begin{bmatrix}
    0 & 0 & \omega & \omega & \omega \\
    0 & 1 & 1 & 1 & 0 \\
    1 & 0 & 0 & 1 & 1 \\
    \hline
    0 & \omega & 0 & \omega & \omega \\
    \omega & \omega & 0 & \omega & 0 \\
    0 & 0 & 1 & 0 & 1 \\
    0 & 0 & 0 & 1 & 1
\end{bmatrix}$ & $\begin{bmatrix}
   1 & 0 & 0 & 0 \\
   1 & 1 & 0 & 0 \\
   0 & 0 & 1 & 1 \\
   0 & 0 & 0 & 1
\end{bmatrix}$ \\
\hline 
& $3$ & \usebox{\circEE} & $1$ & $\begin{bmatrix}
    0 & 0 & \omega & \omega & \omega \\
    0 & 1 & 1 & 1 & 0 \\
    1 & 0 & 0 & 1 & 1 \\
    \hline
    0 & \omega & 0 & \omega & \omega \\
    \omega & \omega & 0 & \omega & 0 \\
    0 & 0 & 1 & 0 & 1 \\
    0 & 0 & 0 & 1 & 1
\end{bmatrix}$ & $\begin{bmatrix}
    1 & 0 & 0 & 0 \\
    0 & 1 & 0 & 0 \\
    1 & 1 & 1 & 0 \\
    1 & 1 & 0 & 1
\end{bmatrix}$ \\ \hline
& $6$ & \usebox{\circFF} & $2$ & $\begin{bmatrix}
    0 & 0 & \omega & \omega & \omega \\
    0 & 1 & 1 & 1 & 0 \\
    1 & 0 & 0 & 1 & 1 \\
    \hline
    0 & \omega & 0 & \omega & \omega \\
    \omega & \omega & 0 & \omega & 0 \\
    0 & 0 & 1 & 0 & 1 \\
    0 & 0 & 0 & 1 & 1
\end{bmatrix}$ & $\begin{bmatrix}
   1 & 0 & 0 & 0 \\
   0 & 1 & 0 & 0 \\
   1 & 1 & 1 & 0 \\
   1 & 0 & 0 & 1
\end{bmatrix}$ \\ \hline 
& $7$ & \usebox{\circGG}  & $15$ & $\begin{bmatrix}
    0 & 0 & \omega & \omega & \omega \\
    0 & 1 & 1 & 1 & 0 \\
    1 & 0 & 0 & 1 & 1 \\
    \hline
    0 & \omega & 0 & \omega & \omega \\
    \omega & \omega & 0 & \omega & 0 \\
    0 & 0 & 1 & 0 & 1 \\
    0 & 0 & 0 & 1 & 1
\end{bmatrix}$ & $\begin{bmatrix}
    1 & 0 & 0 & 0 \\
    1 & 1 & 0 & 0 \\
    0 & 0 & 1 & 1 \\
    1 & 1 & 0 & 1
\end{bmatrix}$ \\ \hline
\end{tabular}

\newsavebox{\circAAA}
\newsavebox{\circBBB}
\newsavebox{\circCCC}
\newsavebox{\circDDD}
\newsavebox{\circEEE}
\newsavebox{\circFFF}
\newsavebox{\circGGG}
\newsavebox{\circHHH}

\savebox{\circAAA}{%
 \begin{quantikz}
\lstick{} & \push{} & \push{  \begin{bsmallmatrix}
    1 & 0 \\
    1 & 1
\end{bsmallmatrix}} & \push{} & \push{} \\
\lstick{} & \push{} & \push{} & \push{} & \push{} \\
\lstick{} & \push{} & \push{  \begin{bsmallmatrix}
    1 & 0 \\
    1 & 1
\end{bsmallmatrix}} & \push{} & \push{} \\
\lstick{} & \push{} & \push{  \begin{bsmallmatrix}
    1 & 0 \\
    1 & 1
\end{bsmallmatrix}} & \push{} & \push{} \\
\lstick{} & \push{} & \push{  \begin{bsmallmatrix}
    1 & 0 \\
    1 & 1
\end{bsmallmatrix}} & \push{} & \push{} 
\end{quantikz}%
}

\savebox{\circBBB}{%
\begin{quantikz}
\lstick{} &  \push{} & \push{} & \permute{3, 2, 4, 1, 5
}  & \push{} \\
\lstick{} & \push{} & \push{} & \push{}  & \push{} \\
\lstick{} & \push{} & \push{} & \push{}  & \push{} \\
\lstick{} & \push{} & \push{} & \push{}  & \push{} \\
\lstick{} & \push{} & \push{} & \push{}  & \push{}
\end{quantikz}%
}

\savebox{\circCCC}{%
 \begin{quantikz}
\lstick{} &  &   \push{ \begin{bsmallmatrix}
    1 & 0 \\
    1 & 1
\end{bsmallmatrix}} & \permute{1,2,4,3,5
} & \\
\lstick{} & & \push{} & \push{} \\
\lstick{} & & \push{} & \push{} \\
\lstick{} & & \push{} & \push{} \\
\lstick{} & & \push{  \begin{bsmallmatrix}
    1 & 0 \\
    1 & 1
\end{bsmallmatrix}} & \push{}
\end{quantikz}%
}

\savebox{\circDDD}{%
 \begin{quantikz}
\lstick{} &  &   \push{ \begin{bsmallmatrix}
    1 & 0 \\
    1 & 1
\end{bsmallmatrix}} & \permute{3, 2, 4, 1, 5
} & \\
\lstick{} & & \push{} & \push{} \\
\lstick{} & & \push{} & \push{}  \\
\lstick{} & & \push{} & \push{} \push{} \\
\lstick{} & & \push{  \begin{bsmallmatrix}
    1 & 0 \\
    1 & 1
\end{bsmallmatrix}} & \push{}
\end{quantikz}%
}

\savebox{\circEEE}{%
 \begin{quantikz}
\lstick{} &  \push{} & \permute{1,4,3,5,2,6
} & \push{} \\
\lstick{} & \push{} & \push{} & \push{} \\
\lstick{} & \push{\begin{bsmallmatrix}
    1 & 1 \\
    1 & 0
\end{bsmallmatrix}} & \push{} & \push{}  \\
\lstick{} & \push{} & \push{} & \push{} \push{} \\
\lstick{} & \push{}  & \push{} & \push{} \\
\lstick{} & \push{} & \push{} & \push{}  \\
\end{quantikz}%
}

\savebox{\circFFF}{%
\begin{quantikz}
\lstick{} & \push{} & \permute{1,2,3,5,4,6
} & \push{} \\
\lstick{} & \push{\begin{bsmallmatrix}
    0 & 1 \\
    1 & 0
\end{bsmallmatrix}} & \push{} & \push{} \\
\lstick{} & \push{\begin{bsmallmatrix}
    0 & 1 \\
    1 & 0
\end{bsmallmatrix}} & \push{} & \push{}  \\
\lstick{} & \push{} & \push{} & \push{} \push{} \\
\lstick{} & \push{} & \push{} & \push{} \\
\lstick{} & \push{\begin{bsmallmatrix}
    1 & 0 \\
    1 & 1
\end{bsmallmatrix}} & \push{} & \push{}  \\
\end{quantikz}%
}

\savebox{\circGGG}{%
\begin{quantikz}
\lstick{} & \push{\begin{bsmallmatrix}
    0 & 1 \\
    1 & 1
\end{bsmallmatrix}} & \push{} & \push{} \\
\lstick{} & \push{\begin{bsmallmatrix}
    1 & 1 \\
    1 & 0
\end{bsmallmatrix}} & \push{} & \push{} \\
\lstick{} & \push{\begin{bsmallmatrix}
    0 & 1 \\
    1 & 1
\end{bsmallmatrix}} & \push{} & \push{} \\
\lstick{} & \push{\begin{bsmallmatrix}
    1 & 1 \\
    1 & 0
\end{bsmallmatrix}} & \push{} & \push{} \\
\lstick{} & \push{\begin{bsmallmatrix}
    0 & 1 \\
    1 & 1
\end{bsmallmatrix}} & \push{} & \push{} \\
\lstick{} & \push{\begin{bsmallmatrix}
    1 & 1 \\
    1 & 0
\end{bsmallmatrix}} & \push{} & \push{} \\
\lstick{} & \push{\begin{bsmallmatrix}
    0 & 1 \\
    1 & 1
\end{bsmallmatrix}} & \push{} & \push{} \\
\end{quantikz}%
}

\savebox{\circHHH}{%
\begin{quantikz}
 \lstick{} & \push{\begin{bsmallmatrix}
     1 & 1 \\
     0 & 1
 \end{bsmallmatrix}} & \push{} & \push{} \\
\lstick{} & \push{\begin{bsmallmatrix}
    1 & 0 \\
    1 & 1
\end{bsmallmatrix}} & \push{} & \push{} \\
\lstick{} & \push{\begin{bsmallmatrix}
    1 & 0 \\
    1 & 1
\end{bsmallmatrix}} & \push{} & \push{} \\
\lstick{} & \push{\begin{bsmallmatrix}
    1 & 0 \\
    1 & 1
\end{bsmallmatrix}} & \push{} & \push{} \\
\lstick{} & \push{\begin{bsmallmatrix}
    1 & 0 \\
    1 & 1
\end{bsmallmatrix}} & \push{} & \push{} \\
\lstick{} & \push{\begin{bsmallmatrix}
   0 & 1 \\
   1 & 0 
\end{bsmallmatrix}} & \push{} & \push{} \\
\lstick{} & \push{\begin{bsmallmatrix}
    0 & 1 \\
    1 & 0
\end{bsmallmatrix}} & \push{} & \push{} \\
\end{quantikz}%
}

\begin{tabular}{|c|c|c|c|c|c|}
 \hline 
 \multicolumn{6}{|c|}{Optimising over controlled-Clifford metric $7|\pi^{ }_{\sigma}| + |\pi^{ }_{e}|$}\\

 \hline
Code & Class & Physical Circuit $\pi$ & $7|\pi_{\sigma}| + |\pi_{e}|$ & $\begin{bmatrix}
    \mathcal{G} \\ \hline \mathcal{B}
\end{bmatrix}$ &  $L(\pi, \mathcal{B})$  \\
\hline
$[[5,2,2]]$ & $2$ & \usebox{\circAAA} & $4$ & $\begin{bmatrix}
    0 & \omega & 0 & 0 & \omega \\
    \omega & 0 & \omega & \omega & \omega \\
    1 & 1 & 1 & 1 & 1 \\ \hline 
    0 & 1 & 0 & 1 & 1 \\
    0 & 1 & 1 & 0 & 1 \\
    0 & 0 & \omega & 0 & \omega \\
    0 & 0 & 0 & \omega & \omega 
\end{bmatrix}$ & $\begin{bmatrix}
    1 & 0 & 0 & 0 \\
    0 & 1 & 0 & 0 \\
    0 & 1 & 1 & 0 \\
    1 & 0 & 0 & 1
\end{bmatrix}$ \\ \hline 
& $4$ & \usebox{\circBBB} & $14$ & $\begin{bmatrix}
    0 & \omega & 0 & 0 & \omega \\
    \omega & 0 & \omega & \omega & \omega \\
    1 & 1 & 1 & 1 & 1 \\ \hline 
    0 & 1 & 0 & 1 & 1 \\
    0 & 1 & 1 & 0 & 1 \\
    0 & 0 & \omega & 0 & \omega \\
    0 & 0 & 0 & \omega & \omega 
\end{bmatrix}$ & $\begin{bmatrix}
    1 & 1 & 0 & 0 \\
    1 & 0 & 0 & 0 \\
    0 & 0 & 0 & 1 \\
    0 & 0 & 1 & 1
\end{bmatrix}$ \\ \hline 
& $6$ & \usebox{\circCCC} & $9$ & $\begin{bmatrix}
    0 & \omega & 0 & 0 & \omega \\
    \omega & 0 & \omega & \omega & \omega \\
    1 & 1 & 1 & \omega^{2} & 1 \\ \hline 
    0 & 1 & 0 & \omega^{2} & 1 \\
    0 & 1 & 1 & 0 & 1 \\
    0 & 0 & \omega & 0 & \omega \\
    0 & 0 & 0 & \omega & \omega 
\end{bmatrix}$ & $\begin{bmatrix}
   0 & 1 & 0 & 0 \\
   1 & 0 & 0 & 0 \\
   1 & 0 & 0 & 1 \\
   0 & 1 & 1 & 0 
\end{bmatrix}$ \\ \hline 
& $9$ & \usebox{\circDDD} & $16$ & $\begin{bmatrix}
    0 & \omega & 0 & 0 & \omega \\
    \omega & 0 & \omega & \omega & \omega \\
    1 & 1 & 1 & \omega^{2} & 1 \\ \hline 
    0 & 1 & 0 & \omega^{2} & 1 \\
    0 & 1 & 1 & 0 & 1 \\
    0 & 0 & \omega & 0 & \omega \\
    0 & 0 & 0 & \omega & \omega 
\end{bmatrix}$ & $\begin{bmatrix}
   1 & 1 & 0 & 0 \\
   1 & 0 & 0 & 0 \\
   1 & 0 & 0 & 1 \\
   1 & 1 & 1 & 1
\end{bmatrix}$ \\ \hline
$[[6,1,3]]$ & $2$ & \usebox{\circEEE} & $15$ & $\begin{bmatrix}
    1 & 0 & 0 & 0 & 0 & \omega \\
    0 & 1 & 1 & \omega & \omega^{2} & 0 \\
    0 & \omega^{2} & 0 & \omega^{2} & \omega^{2} & \omega \\
    0 & 1 & \omega & 0 & \omega & \omega \\
    \omega & \omega^{2} & \omega & \omega & 0 & \omega^{2} \\ \hline 
    0 & 0 & \omega & \omega & \omega & 0 \\
    0 & 0 & 0 & 1 & \omega & \omega
\end{bmatrix}$ & $\begin{bmatrix}
    0 & 1 \\
    1 & 1
\end{bmatrix}$ \\ \hline
& $3$ & \usebox{\circFFF} & $10$ & $\begin{bmatrix}
    1 & 0 & 0 & 0 & 0 & \omega \\
    0 & 1 & 1 & \omega & \omega^{2} & 0 \\
    0 & \omega^{2} & 0 & \omega^{2} & \omega^{2} & \omega \\
    0 & 1 & \omega & 0 & 1 & \omega \\
    \omega & \omega^{2} & \omega & \omega & 0 & \omega^{2} \\ \hline 
    0 & 0 & \omega & \omega & 1 & 0 \\
    0 & 0 & 0 & 1 & 1 & \omega 
\end{bmatrix}$ & $\begin{bmatrix}
    1 & 0 \\
    1 & 1
\end{bmatrix}$ \\ \hline
\end{tabular}

\newpage

\onecolumn

\footnotesize

\begin{tabular}{|c|c|c|c|c|c|}
 \hline 
 \multicolumn{6}{|c|}{Optimising over controlled-Clifford metric $7|\pi_{\sigma}| + |\pi_{e}|$}\\ \hline 
  \hline
Code & Class & Physical Circuit $\pi$ & $7|\pi_{\sigma}|+ |\pi_{e}|$ & $\begin{bmatrix}
    \mathcal{G} \\ \hline \mathcal{B}
\end{bmatrix}$ &  $L(\pi, \mathcal{B})$  \\
\hline

$[[7,1,3]]$ & $2$ & \usebox{\circGGG} & $7$ & $\begin{bmatrix}
    1 & 0 & \omega & 0 & \omega & \omega^{2} & 0 \\
    1 & \omega & \omega & \omega & 0 & 0 & 0 \\
    1 & 0 & 0 & \omega & \omega & 0 & \omega^{2} \\
    \omega & 0 & \omega^{2} & 0 & \omega^{2} & \omega & 0 \\
    \omega & 1 & \omega^{2} & 1 & 0 & 0 & 0 \\
    \omega & 0 & 0 & 1 & \omega^{2} & 0 & 1 \\ \hline 
    0 & \omega & 0 & \omega & 0 & 0 & \omega^{2} \\
    0 & 1 & 0 & 1 & 0 & 0 & 1
\end{bmatrix}$ & $\begin{bmatrix}
    0 & 1 \\
    1 & 1
\end{bmatrix}$ \\ \hline 

 & $3$ & \usebox{\circHHH} & $7$ & $\begin{bmatrix}
    1 & 0 & \omega & 0 & \omega & \omega^{2} & 0 \\
    1 & \omega & \omega & \omega & 0 & 0 & 0 \\
    1 & 0 & 0 & \omega & \omega & 0 & \omega^{2} \\
    \omega & 0 & \omega^{2} & 0 & \omega^{2} & \omega & 0 \\
    \omega & 1 & \omega^{2} & 1 & 0 & 0 & 0 \\
    \omega & 0 & 0 & 1 & \omega^{2} & 0 & 1 \\ \hline 
    0 & \omega & 0 & \omega & 0 & 0 & \omega^{2} \\
    0 & 1 & 0 & 1 & 0 & 0 & 1
\end{bmatrix}$ & $ \begin{bmatrix}
    1 & 1 \\
    0 & 1
\end{bmatrix}$\\ \hline

\end{tabular}

\newpage

\onecolumn

\footnotesize

\newsavebox{\circone}
\newsavebox{\circtwo}
\newsavebox{\circthree}
\newsavebox{\circfour}
\newsavebox{\circfive}
\newsavebox{\circsix}
\newsavebox{\circseven}

\savebox{\circone}{
\begin{quantikz}
\lstick{} & \push{} & \permute{3,2,4,1}  & \push{}  \\
\lstick{} & \push{\begin{bsmallmatrix}
    0 & 1 \\
    1 & 1
\end{bsmallmatrix}} & \push{}  &  \push{} \\
\lstick{} & \push{} & \push{} &  \push{} \\
\lstick{} & \push{} & \push{}  & \push{}
\end{quantikz}
}

\savebox{\circtwo}{
\begin{quantikz}
\lstick{} & \push{} & \push{} & \permute{2,3,4,1} & \push{}  \\
\lstick{} & \push{} & \push{} & \push{} &  \push{} \\
\lstick{} & \push{} & \push{} & \push{} &  \push{} \\
\lstick{} & \push{} & \push{} & \push{} & \push{}
\end{quantikz}
}

\savebox{\circthree}{\begin{quantikz}
\lstick{} & \push{} & \push{} & \permute{1,2,4,3} & \push{}  \\
\lstick{} & \push{} & \push{} & \push{} &  \push{} \\
\lstick{} & \push{} & \push{} & \push{} &  \push{} \\
\lstick{} & \push{} & \push{} & \push{} & \push{}
\end{quantikz}
}

\savebox{\circfour}{\begin{quantikz}
\lstick{} & \push{} & \push{} & \permute{1,3,4,2} & \push{}  \\
\lstick{} & \push{} & \push{} & \push{} &  \push{} \\
\lstick{} & \push{} & \push{} & \push{} &  \push{} \\
\lstick{} & \push{} & \push{} & \push{} & \push{}
\end{quantikz}
}

\savebox{\circfive}{
 \begin{quantikz}
\lstick{} & \push{}  & \permute{2,3,1,4} & \push{} \\
\lstick{} & \push{}  & \push{} &  \push{} \\
\lstick{} & \push{}  & \push{} &  \push{} \\
\lstick{} &  \push{\begin{bsmallmatrix}
0 & 1 \\
1 & 1
\end{bsmallmatrix}}   & \push{} & \push{}
\end{quantikz}
}

\savebox{\circsix}{
\begin{quantikz}
\lstick{} & \push{} & \push{} & \permute{3,4,2,1} & \push{}  \\
\lstick{} & \push{} & \push{} & \push{} &   \push{} \\
\lstick{} & \push{} & \push{} & \push{} &  \push{} \\
\lstick{} & \push{} &  \push{} & \push{} & \push{}
\end{quantikz}
}

\savebox{\circseven}{\begin{quantikz}
\lstick{} & \push{} & \permute{2,3,4,1} & \push{}  \\
\lstick{} & \push{} & \push{} & \push{} &   \\
\lstick{} & \push{  \begin{bsmallmatrix}    0 & 1 \\   1 & 1 \end{bsmallmatrix}} &  \push{} & \push{} \\
\lstick{} & \push{} & \push{} & \push{} &  \\
\end{quantikz}
}

\begin{tabular}{|c|c|c|c|c|c|}
 \hline 
 \multicolumn{6}{|c|}{Optimising over local Clifford metric $|\pi^{ }_{e}|$}\\

 \hline
Code & Class & Physical Circuit $\pi$ & $|\pi_{e}|$ & $\begin{bmatrix}
    \mathcal{G} \\ \hline \mathcal{B}
\end{bmatrix}$ &  $L(\pi, \mathcal{B})$  \\
\hline
$[[4,1,2]]$ & $2$ &  \usebox{\circone} & $1$ & $\begin{bmatrix}
    1 & \omega^{2} & \omega^{2} & 0 \\
    0 & 1 & 1 & \omega^{2} \\
    \omega & 0 & \omega & \omega \\ \hline 
    \omega & 1 & 0 & 0 \\
    0 & \omega & \omega & 0 
\end{bmatrix}$ & $\begin{bmatrix}
    0 & 1 \\
    1 & 1
\end{bmatrix}$ \\
\hline 

& $3$ & \usebox{\circtwo} & $0$  & $\begin{bmatrix}
    1 & \omega^{2} & \omega & 0 \\
    0 & 1 & \omega^{2} & \omega \\
    \omega & 0 & 1 & \omega^{2} \\ \hline 
    \omega & 1 & 0 & 0 \\
    0 & \omega & 1 & 0 
\end{bmatrix}$ & $\begin{bmatrix}
    0 & 1 \\
    1 & 0 
\end{bmatrix}$ \\ \hline

$[[4,2,2]]$ & $2$ & \usebox{\circthree} & $0$ & $\begin{bmatrix}
    1 & 1 & 1 & 1 \\
    \omega & \omega & \omega & \omega \\ \hline 
    1 & 1 & 0 & 0 \\
    \omega & \omega & 0 & 0 \\
    \omega & 0 & \omega & 0 \\
    1 & 0 & 1 & 0
\end{bmatrix}$ & $\begin{bmatrix}
   1 & 0 & 0 & 1 \\
   0 & 1 & 1 & 0 \\
   0 & 0 & 1 & 0 \\
   0 & 0 & 0 & 1 
\end{bmatrix}$\\ \hline 

& $4$ & \usebox{\circfour} & $0$ & $\begin{bmatrix}
    1 & 1 & 1 & 1 \\
    \omega & \omega & \omega & \omega \\ \hline 
    1 & 1 & 0 & 0 \\
    \omega & \omega & 0 & 0 \\
    \omega & 0 & \omega & 0 \\
    1 & 0 & 1 & 0 
\end{bmatrix}$ & $\begin{bmatrix}
   0 & 0 & 0 & 1 \\
   0 & 0 & 1 & 0 \\
   0 & 1 & 1 & 0 \\
   1 & 0 & 0 & 1 
\end{bmatrix}$ \\ \hline 

& $5$ & \usebox{\circfive} & $1$ & $\begin{bmatrix}
    \omega^{2} & \omega & 1 & 1 \\
    1 & \omega^{2} & \omega & \omega \\ \hline 
    \omega^{2} & \omega & 0 & 0 \\
    1 & \omega^{2} & 0 & 0 \\
    1 & 0 & \omega & 0 \\
    \omega^{2} & 0 & 1 & 0 
\end{bmatrix}$ & $\begin{bmatrix}
0 & 1 & 1 & 0 \\
1 & 1 & 1 & 1 \\
1 & 1 & 0 & 0 \\
0 & 1 & 0 & 0 
\end{bmatrix}$ \\ \hline

& $6$ &  \usebox{\circsix} & $0$ & $\begin{bmatrix}
    1 & 1 & 1 & 1 \\
    \omega^{2} & \omega^{2} & \omega & \omega \\ \hline 
    1 & 1 & 0 & 0 \\
    \omega^{2} & \omega^{2} & 0 & 0 \\
    \omega^{2} & 0 & \omega & 0 \\
    1 & 0 & 1 & 0
\end{bmatrix}$ & $\begin{bmatrix}
1 & 1 & 1 & 1 \\
0 & 1 & 1 & 0 \\
0 & 0 & 1 & 0 \\
0 & 0 & 1 & 1 
\end{bmatrix}$ \\ \hline

& $9$ &  \usebox{\circseven} & $1$ & $\begin{bmatrix}
    \omega^{2} & \omega & 1 & 1 \\
    1 & \omega^{2} & \omega & \omega \\ \hline 
    \omega^{2} & \omega & 0 & 0 \\
    1 & \omega^{2} & 0 & 0 \\
    1 & 0 & \omega & 0 \\
    \omega^{2} & 0 & 1 & 0 
\end{bmatrix}$ & $\begin{bmatrix}
0 & 1 & 0 & 0 \\
1 & 1 & 0 & 0 \\
1 & 1 & 1 & 1 \\
0 & 1 & 1 & 0 
\end{bmatrix}$ \\ \hline

\end{tabular}

\newsavebox{\circoneone}
\newsavebox{\circtwotwo}
\newsavebox{\circthreethree}
\newsavebox{\circfourfour}
\newsavebox{\circfivefive}
\newsavebox{\circsixsix}
\newsavebox{\circsevenseven}

\savebox{\circoneone}{
\begin{quantikz}
\lstick{} & \push{} & \permute{2,3,5,4,1} & \push{}  \\
\lstick{} & \push{} & \push{} & \push{} &   \\
\lstick{} & \push{} & \push{} & \push{} &  \\
\lstick{} & \push{\begin{bsmallmatrix}
    0 & 1 \\
    1 & 0
\end{bsmallmatrix}} & \push{} &  \push{} \\
\lstick{} & \push{} & \push{} & \push{} &  
\end{quantikz}
}

\savebox{\circtwotwo}{
\begin{quantikz}
\lstick{} & \push{} & \permute{3,2,4,1,5} & \push{}  \\
\lstick{} & \push{\begin{bsmallmatrix}
    0 & 1 \\
    1 & 1
\end{bsmallmatrix}} & \push{} &  \push{} \\
\lstick{} & \push{} & \push{} & \push{} &   \\
\lstick{} & \push{} & \push{} & \push{} &  \\
\lstick{} & \push{} & \push{} & \push{} &  
\end{quantikz}
}

\savebox{\circthreethree}{
\begin{quantikz}
\lstick{} & \push{} & \permute{5,3,1,4,2} & \push{}  \\
\lstick{} & \push{} & \push{} & \push{} &   \\
\lstick{} & \push{} & \push{} & \push{} &  \\
\lstick{} & \push{\begin{bsmallmatrix}
    1 & 1 \\
    0 & 1
\end{bsmallmatrix}} & \push{} &  \push{} \\
\lstick{} & \push{} & \push{} & \push{} &  
\end{quantikz}
}

\savebox{\circfourfour}{
\begin{quantikz}
\lstick{} & \push{}  & \push{} & \permute{2,1,5,4,3} & \push{}   \\
\lstick{} & \push{} & \push{} & \push{} &  \push{}  \\
\lstick{} & \push{} & \push{} & \push{} &  \push{} \\
\lstick{} & \push{} & \push{} &  \push{} & \push{} \\
\lstick{} & \push{} & \push{} & \push{} &  \push{} 
\end{quantikz}
}

\savebox{\circfivefive}{
\begin{quantikz}
\lstick{} & \push{} & \push{} & \push{} & \push{}  \\
\lstick{} & \push{\begin{bsmallmatrix}
    1 & 1 \\
    0 & 1
\end{bsmallmatrix}} & \push{} &  \push{} & \push{} \\
\lstick{} & \push{} & \push{} & \push{} &   \push{} \\
\lstick{} & \push{} & \push{} & \push{} & \push{}  \\
\lstick{} & \push{} & \push{} & \push{} &  \push{}
\end{quantikz}
}

\savebox{\circsixsix}{
\begin{quantikz}
\lstick{} & \push{\begin{bsmallmatrix}
    1 & 1 \\
    0 & 1
\end{bsmallmatrix}} & \push{} & \push{} & \push{} \\
\lstick{} & \push{\begin{bsmallmatrix}
    1 & 1 \\
    0 & 1
\end{bsmallmatrix}} & \push{} & \push{} & \push{}  \\
\lstick{} & \push{} & \push{} & \push{} & \push{} \\
\lstick{} & \push{} & \push{} &  \push{} & \push{} \\
\lstick{} & \push{} & \push{} & \push{} &  \push{} 
\end{quantikz}
}

\savebox{\circsevenseven}{
 \begin{quantikz}
\lstick{} & \push{}  & \permute{2,1,5,4,3} & \push{}  \\
\lstick{} & \push{\begin{bsmallmatrix}
    1 & 1 \\
    0 & 1
\end{bsmallmatrix}}  & \push{} & \push{} &   \\
\lstick{} & \push{}  & \push{} & \push{} \\
\lstick{} & \push{}  &  \push{} & \push{} \\
\lstick{} & \push{}  & \push{} & \push{} 
\end{quantikz}
}

\begin{tabular}{|c|c|c|c|c|c|}
 \hline 
 \multicolumn{6}{|c|}{Optimising over local Clifford metric $|\pi^{ }_{e}|$}\\

 \hline
Code & Class & Physical Circuit $\pi$ & $ |\pi_{e}|$ & $\begin{bmatrix}
    \mathcal{G} \\ \hline \mathcal{B}
\end{bmatrix}$ &  $L(\pi, \mathcal{B})$  \\
\hline
$[[5,1,2]]$ & $3$ & \usebox{\circoneone} & $1$ & $\begin{bmatrix}
    \omega & 1 & 0 & \omega & 0 \\
    0 & 0 & \omega & \omega & 1 \\
    0 & \omega & 1 & 1 & 0 \\
    1 & 0 & 0 & 1 & \omega \\ \hline 
    0 & 1 & 0 & \omega & 1 \\
    0 & 0 & 1 & 0 & \omega
\end{bmatrix}$&  $\begin{bmatrix}
    0 & 1 \\ 1 & 0
\end{bmatrix}$ \\ \hline

$[[5,1,3]]$ & $2$ & \usebox{\circtwotwo} & $1$ & $\begin{bmatrix}
    \omega^{2} & \omega & \omega & 1 & 0 \\
    \omega & 0 & \omega & \omega & 1 \\
    0 & \omega & 1 & \omega & \omega^{2} \\
    \omega & \omega^{2} & 0 & 1 & \omega^{2} \\ \hline 
    1 & \omega & 1 & 0 & 0 \\
    \omega & \omega^{2} & \omega & 0 & 0 
\end{bmatrix}$ & $\begin{bmatrix}
    1 & 1 \\
    1 & 0 
\end{bmatrix}$ \\ \hline

& $3$ & \usebox{\circthreethree} & $1$ & $\begin{bmatrix}
    \omega^{2} & 1 & \omega & 1 & 0 \\
    1 & 0 & \omega & \omega^{2} & 1 \\
    0 & 1 & 1 & \omega^{2} & \omega \\
    1 & \omega^{2} & 0 & 1 & \omega \\ \hline 
    \omega & 1 & 1 & 0 & 0 \\
    1 & \omega^{2} & \omega & 0 & 0
\end{bmatrix}$ & $\begin{bmatrix}
    1 & 0 \\
    1 & 1
\end{bmatrix}$\\ \hline

$[[5,2,1]]$ & $2$ & \usebox{\circfourfour} & $0$ & $\begin{bmatrix}
    0 & 0 & \omega & \omega & \omega \\
    0 & 1 & 1 & 1 & 0 \\
    1 & 0 & 0 & 1 & 1 \\
    \hline 
    0 & \omega & 0 & \omega & \omega \\
    \omega & \omega & 0 & \omega & 0 \\
    0 & 0 & 1 & 0 & 1 \\
    0 & 0 & 0 & 1 & 1
\end{bmatrix}$ & $\begin{bmatrix}
1 & 0 & 0 & 0 \\
1 & 1 & 0 & 0 \\
0 & 0 & 1 & 1 \\
0 & 0 & 0 & 1
\end{bmatrix}$ \\ \hline

& $3$ & \usebox{\circfivefive} & $1$ & $\begin{bmatrix}
    0 & 0 & \omega & \omega & \omega \\
    0 & 1 & 1 & 1 & 0 \\
    1 & 0 & 0 & 1 & 1 \\
    \hline 
    0 & \omega & 0 & \omega & \omega \\
    \omega & \omega & 0 & \omega & 0 \\
    0 & 0 & 1 & 0 & 1\\
    0 & 0 & 0 & 1 & 1
\end{bmatrix}$ & $\begin{bmatrix}
1 & 0 & 0 & 0 \\
0 & 1 & 0 & 0 \\
1 & 1 & 1 & 0 \\
1 & 1 & 0 & 1 
\end{bmatrix}$ \\ \hline

& $6$ &  \usebox{\circsixsix} & $2$ & $\begin{bmatrix}
    0 & 0 & \omega & \omega & \omega \\
    0 & 1 & 1 & 1 & 0 \\
    1 & 0 & 0 & 1 & 1 \\
    \hline 
    0 & \omega & 0 & \omega & \omega \\
    \omega & \omega & 0 & \omega & 0 \\
    0 & 0 & 1 & 0 & 1 \\
    0 & 0 & 0 & 1 & 1
\end{bmatrix}$ & $\begin{bmatrix}
1 & 0 & 0 & 0 \\
0 & 1 & 0 & 0 \\
1 & 1 & 1 & 0 \\
1 & 0 & 0 & 1
\end{bmatrix}$ \\ \hline

& $7$ & \usebox{\circsevenseven} & $1$ & $\begin{bmatrix}
    0 & 0 & \omega & \omega & \omega \\
    0 & 1 & 1 & 1 & 0 \\
    1 & 0 & 0 & 1 & 1 \\
    \hline 
    0 & \omega & 0 & \omega & \omega \\
    \omega & \omega & 0 & \omega & 0 \\
    0 & 0 & 1 & 0 & 1 \\
    0 & 0 & 0 & 1 & 1
\end{bmatrix}$ & $\begin{bmatrix}
1 & 0 & 0 & 0 \\
1 & 1 & 0 & 0 \\
0 & 0 & 1 & 1 \\
1 & 1 & 0 & 1
\end{bmatrix}$ \\ \hline

\end{tabular}

\newpage

\newsavebox{\circoneoneone}
\newsavebox{\circtwotwotwo}
\newsavebox{\circthreethreethree}
\newsavebox{\circfourfourfour}
\newsavebox{\circfivefivefive}
\newsavebox{\circsixsixsix}
\newsavebox{\circsevensevenseven}
\newsavebox{\circeight}

\savebox{\circoneoneone}{
\begin{quantikz}
\lstick{} & \push{} & \push{}  & \permute{1,2,4,3,5} & \push{}  \\
\lstick{} & \push{} & \push{} & \push{} & \push{} \\
\lstick{} & \push{} & \push{} & \push{} & \push{} \\
\lstick{} & \push{} & \push{} &  \push{} & \push{} \\
\lstick{} & \push{} & \push{} & \push{} & \push{} 
\end{quantikz}
}

\savebox{\circtwotwotwo}{
\begin{quantikz}
\lstick{} & \push{} & \push{} & \permute{3,2,4,1,5} & \push{}  \\
\lstick{} & \push{} & \push{} & \push{} & \push{} \\
\lstick{} & \push{} & \push{} & \push{} & \push{} \\
\lstick{} & \push{} & \push{} &  \push{} & \push{} \\
\lstick{} & \push{} & \push{} & \push{} & \push{} 
\end{quantikz}
}

\savebox{\circthreethreethree}{
\begin{quantikz}
\lstick{} & \push{\begin{bsmallmatrix}
1 & 0 \\
1 & 1
\end{bsmallmatrix}} & \permute{1,2,4,3,5} & \push{}  \\
\lstick{} & \push{} & \push{} & \push{} \\
\lstick{} & \push{} & \push{} & \push{} \\
\lstick{} & \push{} &  \push{} & \push{} \\
\lstick{} & \push{\begin{bsmallmatrix}
    1 & 0 \\
    1 & 1
\end{bsmallmatrix}} & \push{} & \push{} 
\end{quantikz}
}

\savebox{\circfourfourfour}{
\begin{quantikz}
\lstick{} &  & \push{\begin{bsmallmatrix}
    1 & 0 \\
    1 & 1 
\end{bsmallmatrix}} & \permute{3,2,4,1,5} & \push{}  \\
\lstick{} &  & \push{} & \push{} & \push{} \\
\lstick{} &  & \push{} & \push{} & \push{} \\
\lstick{} &  & \push{} &  \push{} & \push{} \\
\lstick{} &  & \push{\begin{bsmallmatrix}
    1 & 0 \\
    1 & 1
\end{bsmallmatrix}} & \push{} & \push{} 
\end{quantikz}
}

\savebox{\circfivefivefive}{
\begin{quantikz}
\lstick{} & & \push{} & \permute{1,4,3,5,2,6} & \push{}  \\
\lstick{} & & \push{} & \push{} & \push{} \\
\lstick{} & & \push{\begin{bsmallmatrix}
    1 & 1 \\
    1 & 0
\end{bsmallmatrix}} & \push{} & \push{} \\
\lstick{} & & \push{} &  \push{} & \push{} \\
\lstick{} & & \push{} & \push{} & \push{} \\
\lstick{} & & \push{} & \push{} & \push{} 
\end{quantikz}
}

\savebox{\circsixsixsix}{
\begin{quantikz}
\lstick{} & & \push{} & \permute{1,4,5,3,2,6} & \push{}  \\
\lstick{} & & \push{} & \push{} & \push{} \\
\lstick{} & & \push{} & \push{} & \push{} \\
\lstick{} & & \push{} & \push{} & \push{} \\
\lstick{} & & \push{} & \push{} & \push{} \\
\lstick{} & & \push{\begin{bsmallmatrix}
    1 & 0 \\
    1 & 1
\end{bsmallmatrix}} & \push{} & \push{} 
\end{quantikz}
}

\savebox{\circsevensevenseven}{
\begin{quantikz}
\lstick{} & \push{} & \push{\permute{7,4,5,1,6,2,3}} & \push{} \\    
\lstick{} & \push{} & \push{} & \push{} \\  
\lstick{} & \push{\begin{bsmallmatrix}
          0 & 1 \\
          1 & 0 
      \end{bsmallmatrix}} & \push{} & \push{} \\
\lstick{} & \push{} & \push{} & \push{} \\
\lstick{} & \push{} & \push{} & \push{} \\         
 \lstick{} & \push{\begin{bsmallmatrix}
     1 & 0 \\
     1 & 1
 \end{bsmallmatrix}} & \push{} & \push{} \\
 \lstick{} & \push{} & \push{} & \push{} \\
\end{quantikz}
}

\savebox{\circeight}{
\begin{quantikz}
\lstick{} & \push{} & \push{\permute{2,1,6,7,5,4,3}} & \push{} \\
\lstick{} & \push{} & \push{} & \push{} \\
\lstick{} & \push{} & \push{} & \push{} \\
\lstick{} & \push{} & \push{} & \push{} \\
\lstick{} & \push{} & \push{} & \push{} \\
 \lstick{} & \push{\begin{bsmallmatrix}
     1 & 1 \\
     0 & 1
 \end{bsmallmatrix}} & \push{} & \push{} \\
 \lstick{} & \push{} & \push{} & \push{} \\
\end{quantikz}
}

\begin{tabular}{|c|c|c|c|c|c|} 
 \hline 
 \multicolumn{6}{|c|}{Optimising over local Clifford metric $|\pi^{ }_{e}|$}\\

 \hline
Code & Class & Physical Circuit $\pi$ & $ |\pi_{e}|$ & $\begin{bmatrix}
    \mathcal{G} \\ \hline \mathcal{B}
\end{bmatrix}$ &  $L(\pi, \mathcal{B})$  \\
\hline

$[[5,2,2]]$ & $2$ & \usebox{\circoneoneone} & $0$ & $\begin{bmatrix}
    0 & \omega & 0 & 0 & \omega \\
    \omega & 0 & \omega & \omega & \omega \\ 
    1 & 1 & 1 & 1 & 1 \\
     \hline 
     0 & 1 & 0 & 1 & 1 \\
     0 & 1 & 1 & 0 & 1 \\
     0 & 0 & \omega & 0 & \omega \\
     0 & 0 & 0 & \omega & \omega 
\end{bmatrix}$ &  $\begin{bmatrix}
    0 & 1 & 0 & 0 \\
    1 & 0 & 0 & 0 \\
    0 & 0 & 0 & 1 \\
    0 & 0 & 1 & 0
\end{bmatrix}$\\ \hline

& $4$ & \usebox{\circtwotwotwo} & $0$ & $\begin{bmatrix}
    0 & \omega & 0 & 0 & \omega \\
    \omega & 0 & \omega & \omega & \omega \\
    1 & 1 & 1 & 1 & 1 \\ \hline 
    0 & 1 & 0 & 1 & 1\\
    0 & 1 & 1 & 0 & 1 \\
    0 & 0 & \omega & 0 & \omega \\
    0 & 0 & 0 & \omega & \omega 
\end{bmatrix}$ & $\begin{bmatrix}
    1 & 1 & 0 & 0 \\
    1 & 0 & 0 & 0 \\
    0 & 0 & 0 & 1 \\
    0 & 0 & 1 & 1
\end{bmatrix}$\\ \hline

& $6$ & \usebox{\circthreethreethree} & $2$ & $\begin{bmatrix}
    0 & \omega & 0 & 0 & \omega \\
    \omega & 0 & \omega & \omega & \omega \\
    1 & 1 & 1 & \omega^{2} & 1 \\ \hline 
    0 & 1 & 0 & \omega^{2} & 1 \\
    0 & 1 & 1 & 0 & 1 \\
    0 & 0 & \omega & 0 & \omega \\
    0 & 0 & 0 & \omega & \omega 
\end{bmatrix}$ & $\begin{bmatrix}
0 & 1 & 0 & 0 \\
1 & 0 & 0 & 0 \\
1 & 0 & 0 & 1 \\
0 & 1 & 1 & 0 
\end{bmatrix}$ \\ \hline

& $9$ & \usebox{\circfourfourfour} & $2$ & $\begin{bmatrix}
    0 & \omega & 0 & 0 & \omega \\
    \omega & 0 & \omega & \omega & \omega \\
    1 & 1 & 1 & \omega^{2} & 1 \\ \hline 
    0 & 1 & 0 & \omega^{2} & 1 \\
    0 & 1 & 1 & 0 & 1 \\
    0 & 0 & \omega & 0 & \omega \\
    0 & 0 & 0 & \omega & \omega
\end{bmatrix}$ & $\begin{bmatrix}
1 & 1 & 0 & 0 \\
1 & 0 & 0 & 0 \\
1 & 0 & 0 & 1 \\
1 & 1 & 1 & 1
\end{bmatrix}$\\ \hline

$[[6,1,3]]$ & $2$ & \usebox{\circfivefivefive} & $1$ & $\begin{bmatrix}
    1 & 0 & 0 & 0 & 0 & \omega \\
    0 & 1 & 1 & \omega & \omega^{2} & 0 \\
    0 & \omega^{2} & 0 & \omega^{2} & \omega^{2} & \omega \\
    0 & 1 & \omega & 0 & \omega & \omega \\
    \omega & \omega^{2} & \omega & \omega & 0 & \omega^{2} \\ \hline 
    0 & 0 & \omega & \omega & \omega & 0 \\
    0 & 0 & 0 & 1 & \omega & \omega
\end{bmatrix}$ & $\begin{bmatrix}
    0 & 1 \\
    1 & 1
\end{bmatrix}$  \\ \hline 

& $3$ & \usebox{\circsixsixsix} & $1$ & $\begin{bmatrix}
    1 & 0 & 0 & 0 & 0 & \omega \\
    0 & 1 & 1 & \omega & \omega & 0 \\
    0 & \omega^{2} & 0 & 1 & \omega & \omega \\
    0 & 1 & \omega & 0 & \omega^{2} & \omega \\
    \omega & \omega^{2} & \omega & \omega & 0 & \omega^{2} \\ \hline 
    0 & 0 & \omega & \omega & \omega^{2} & 0 \\
    0 & 0 & 0 & \omega^{2} & \omega^{2} & \omega
\end{bmatrix}$ & $\begin{bmatrix}
    1 & 0 \\
    1 & 1
\end{bmatrix}$ \\ \hline

\end{tabular}

\newpage

\begin{tabular}{|c|c|c|c|c|c|} 
 \hline 
 \multicolumn{6}{|c|}{Optimising over local Clifford metric $|\pi^{ }_{e}|$}\\

 \hline
Code & Class & Physical Circuit $\pi$ & $ |\pi_{e}|$ & $\begin{bmatrix}
    \mathcal{G} \\ \hline \mathcal{B}
\end{bmatrix}$ &  $L(\pi, \mathcal{B})$  \\
\hline

$[[7,1,3]]$ & $2$ & \usebox{\circsevensevenseven} & $2$ & $\begin{bmatrix}
    1 & 0 & \omega & 0 & \omega & \omega^{2} & 0 \\
    1 & \omega^{2} & \omega & 0 & 0 & 0 & \omega^{2}\\
    1 & \omega^{2} & 0 & \omega & \omega & 0 & 0 \\
    \omega & 0 & \omega^{2} & 0 & 1 & \omega & 0 \\
    \omega & 1 & \omega^{2} & 0 & 0 & 0 & 1 \\
    \omega & 1 & 0 & \omega^{2} & 1 & 0 & 0 \\ \hline 
    0 & \omega^{2} & 0 & \omega & 0 & 0 & \omega^{2} \\
    0 & 1 & 0 & \omega^{2} & 0 & 0 & 1
\end{bmatrix}$ & $\begin{bmatrix}
    0 & 1 \\
    1 & 1 
\end{bmatrix}$ \\ \hline

& $3$ & \usebox{\circeight} & $1$ & $\begin{bmatrix}
    1 & 0 & \omega & 0 & \omega & \omega^{2} & 0 \\
    1 & \omega & \omega & \omega & 0 & 0 & 0 \\
    1 & 0 & 0 & \omega & \omega & 0 & \omega^{2} \\
    \omega & 0 & \omega^{2} & 0 & \omega^{2} & \omega & 0 \\
    \omega & 1 & \omega^{2} & \omega^{2} & 0 & 0 & 0 \\
    \omega & 0 & 0 & \omega^{2} & \omega^{2} & 0 & \omega \\ \hline 
    0 & \omega & 0 & \omega & 0 & 0 & \omega^{2} \\
    0 & 1 & 0 & \omega^{2} & 0 & 0 & \omega 
    
\end{bmatrix}$ & $\begin{bmatrix}
    0 & 1 \\
    1 & 0 
\end{bmatrix}$ \\ \hline

\end{tabular}

\section{Discussion}
\normalsize

The framework introduced in this work provides an exhaustive method for enumerating and optimising all automorphism-induced logical operations for small stabiliser codes. By combining symplectic analysis, code equivalence, and independent optimisation of SWAP and Clifford resources, we have demonstrated that logically identical operations can yield significantly different physical realisations. This highlights how algebraic symmetries can be exploited to identify circuits that are not only fault-tolerant but also resource-efficient in realistic implementations.

From a theoretical standpoint, Theorems \ref{thm 1} and \ref{thm transversal} fully characterise the physical Clifford circuits enacting logical Clifford gates incorporating choice of logical basis and choice of ``version'' of a stabiliser code. These insights have the benefit of cutting down the search space, so that we can exhaustively optimise over both the aforementioned degrees of freedom simultaneously. We applied our framework to give optimal circuits with respect to two particular metrics that we have motivated in the introduction, although other authors may prefer to optimise alternative metrics motivated by different tasks or physical constraints.

While the current numerical results are restricted to small codes with $n \leq 7$ and logical dimensions $k \leq 2$, the methods generalise naturally to larger codes. The principal limitation for larger codes ($n >7$ or $k>2$) arises from the combinatorial growth of both the automorphism group and the symplectic group $Sp(2k,2)$. Future work will focus on developing algorithmic approaches that remain tractable for larger and higher-dimensional codes.

Finally, developing an open-access repository based on this framework, building on tools such as \href{https://qecdb.org/}{\textbf{QECDB}} \cite{QECDB} and \href{https://codetables.de/}{\textbf{codetables.de}} \cite{Codetable}, and perhaps a software library in the vein of autqec \cite{autqec_code} would enable systematic exploration of code equivalence classes, automate Clifford decomposition and support the experimental design of near-term fault-tolerant quantum devices.

\section{Acknowledgements}
A.M.~and M.H~acknowledge funding from the Royal Society - Research Ireland (URF). A.M.~also acknowledges funding from the College of Science and Engineering at the University of Galway. M.H~thanks Mark Webster for interesting discussions on code automorphisms at the Mathematical Foundations of Quantum Advantage workshop held at Simon Fraser University funded by the Pacific Institute for the Mathematical Sciences (PIMS).

\bibliographystyle{plainurl}
\bibliography{biblio}

\onecolumn\newpage
\appendix

\section{Proofs of Propositions}

\subsection{Proof of Proposition \ref{eqn logical action}}\label{appx L(pi,B)}
\begin{proof}
    
We have defined $\mathcal{C}$ to be an additive $GF(4)$ code, meaning that its dual $\mathcal{C}^{\perp}$ can be defined with respect to the trace inner product (as described in section \ref{subsection logical action}), \begin{align}\mathcal{C}^{\perp} = \{ u \in GF(4)^{n}: Tr(u . \overline{c}) = 0 \text{ } \forall \text{ } c \in \mathcal{C} \}.\end{align}
The logical basis $\mathcal{B} \in \mathcal{C}^{\perp} \backslash \mathcal{C} $ consists of elements in $GF(4)^{n}$ which correspond to the logical $\mathcal{X}$ and logical $\mathcal{Z}$ type operators in $N(S) \backslash S$. Furthermore, $\Gamma(\mathcal{C}) \leq Ham(\mathcal{C})$ ensures that $\pi \in \Gamma(\mathcal{C})$ preserves trace inner products and as such preserves elements of $C^{\perp}$ \cite{calderbank1997quantum,roetteler}.
This means that given the rows of $\mathcal{B}$ are elements of $ \mathcal{C}^{\perp}\backslash \mathcal{C}$ and $\pi \in \Gamma(\mathcal{C})$, then the rows of the transformed logical basis, $\pi(\mathcal{B})$, will also be in $\mathcal{C}^{\perp} \backslash \mathcal{C}$. The preservation of $\mathcal{C}^{\perp}$ by $\Gamma(\mathcal{C})$ demonstrates that the action of $\pi$ corresponds to that of a logical Clifford operator, mapping elements corresponding to logical-$\mathcal{X}$ and logical-$\mathcal{Z}$ operators, either to themselves, or to other valid logical Pauli operators. 

We can calculate the logical operation implemented by $\pi$ via $\mathcal{D}_{\mathcal{B}} \odot \pi(\mathcal{B})^{T}$, where $$\mathcal{D}_{\mathcal{B}} = \begin{bmatrix}
    \mathcal{Z}_{L}^{(1)} \\
    \vdots \\
    \mathcal{Z}_{L}^{(k)} \\
    \mathcal{X}_{L}^{(1)} \\
    \vdots \\
    \mathcal{X}_{L}^{(k)}
\end{bmatrix},$$ and given $\pi(\mathcal{X}_{L}^{(1)}) = \mathcal{X}_{L}'^{(1)} $, we can construct the transformed basis,  
$$\pi(\mathcal{B})^{T} = \begin{bmatrix}
    \mathcal{X}_{L}'^{(1)} & 
    \dots 
    \mathcal{X}_{L}'^{(k)} &
    \mathcal{Z}_{L}'^{(1)} &
    \dots &
    \mathcal{Z}_{L}'^{(k)}
\end{bmatrix}.$$ 
 The operation $\odot$ is applied between the dual basis $\mathcal{D}_{\mathcal{B}}$ and the transformed basis $\pi(\mathcal{B})^{T}$ such that: 

\begin{align}\label{logical action odot rowwise}  \mathcal{D}_{\mathcal{B}} \odot \pi(\mathcal{B})^{T}  = \begin{bmatrix}
    Tr(\mathcal{Z}_{L}^{(1)} \cdot \overline{\mathcal{X}_{L}'^{(1)}}) & \dots & Tr(\mathcal{Z}_{L}^{(1)} \cdot \overline{\mathcal{Z}_{L}'^{(k)}}) \\
    \vdots & \vdots & \vdots \\ 
     Tr(\mathcal{X}_{L}^{(k)} \cdot \overline{\mathcal{X}_{L}'^{(1)}}) & \dots & Tr(\mathcal{X}_{L}^{(k)} \cdot \overline{\mathcal{Z}_{L}'^{(k)}})
\end{bmatrix} \end{align}

As we can see from Eq.~\eqref{logical action odot rowwise}, $(\mathcal{D}_{\mathcal{B}} \odot \pi(\mathcal{B})^{T})_{i,j} =1$ tells us exactly where $\pi(\mathcal{B}_{j})$ anti-commutes with 
\begin{enumerate}[i]
    \item $\mathcal{B}_{i+k}$, if $i \leq k$ 
   \item $\mathcal{B}_{i-k}$ if $i>k$.
\end{enumerate}

As such, row $i$ of $ \mathcal{D}_{\mathcal{B}} \odot \pi(\mathcal{B})^{T}$, gives exactly the action of $\pi(\mathcal{B}_{i})$, meaning that we can express the logical action of $\pi$ acting on the encoded space as shown in Proposition \ref{eqn logical action}, $$L(\pi,\mathcal{B}) = \mathcal{D}_{\mathcal{B}} \odot  \pi(\mathcal{B})^{T}.$$

\end{proof}

\subsection{Proofs of Proposition \ref{eqn L=ALAinv} and \ref{prop code equivalence no effect}}

For this section of proofs, we rely on the following facts:
\begin{itemize}
    \item For $F \in Sp(2k,2)$ and $\Omega_{r} = \begin{bmatrix}
        0 & \mathbb{I}_{r} \\
        \mathbb{I}_{r} & 0
    \end{bmatrix}$,
    \begin{enumerate}
    \item\label{fact 1} $F \Omega F^{T} = \Omega$
    \item\label{fact 2} $F^{-1} = \Omega F^{T} \Omega$
    \item\label{fact 3} $\Omega = \Omega^{-1} = \Omega^{T}$
    \item\label{fact 4} $\Omega^{2} = I$
\end{enumerate}
 \item For $\mathcal{W}, \mathcal{V} \in GF(4)^{2k \times n}$,
 \begin{enumerate}[a.]
     \item\label{fact 5} $\mathcal{W} \odot \mathcal{V}^{T} = \phi(\mathcal{W}) \Omega_{2n} \phi(\mathcal{V})^{T} = W\Omega_{2n}V^{T}$
 \end{enumerate}
\end{itemize}

\subsubsection{Proof of Proposition \ref{eqn L=ALAinv}}\label{appx ALA}
\begin{proof}
By Proposition \ref{eqn logical action}, $L(\pi, A\mathcal{B}) = \mathcal{D}_{A\mathcal{B}} \odot \pi(A\mathcal{B})^{T}$, for $A \in Sp(2k,2)$. We use the five facts listed at the beginning of this section to prove that $L(\pi,A\mathcal{B}) = (A^{-1})^{T}(L(\pi, \mathcal{B}))A^{T}$

$$L(\pi, A\mathcal{B}) = \mathcal{D}_{A\mathcal{B}} \odot \pi(A\mathcal{B})^{T} $$

$$= \underbrace{(\Omega_{2k}A\mathcal{B})}_{Eq.~\eqref{eqn destab}} \odot \pi(A\mathcal{B})^{T}$$ $$\underbrace{=}_{Fact \text{ } a.} (\Omega_{2k}A B) \Omega_{2n} (\underbrace{A\mathcal{B} F_{\pi}^{T}}_{Eq.~\eqref{eqn alpha:m pi to Fpi}})^{T} $$
$$= \Omega_{2k}A\underbrace{\Omega_{2k}\Omega_{2k}}_{Fact \text{ } 4. } \underbrace{D_{B}}_{Eq.~\eqref{eqn destab}} \Omega_{2n} F_{\pi}B^{T}A^{T}$$
$$\underbrace{=}_{Fact \text{ } a.} \underbrace{(A^{-1})^{T}}_{Fact \text{ } 2.}\underbrace{\mathcal{D}_{\mathcal{B}}}_{Eq.~\eqref{eqn destab}} \odot \underbrace{\pi(\mathcal{B})^{T}}_{Eq.~\eqref{eqn alpha:m pi to Fpi}}A^{T}$$
$$= (A^{-1})^{T} \underbrace{L(\pi, \mathcal{B})}_{Eq.~\eqref{eqn logical action}}A^{T}$$

 This proves that $$L(\pi, A\mathcal{B})=(A^{-1})^{T}L(\pi,\mathcal{B})A^{T}.$$
\end{proof}

\subsubsection{Proof of Proposition \ref{prop code equivalence no effect}}\label{appendix L(id)=L(tau)}

\begin{proof}
    
 Recall that by Eq.~\eqref{eqn alpha:m pi to Fpi}, $\alpha(\pi) = F_{\pi}$ is a homomorphism. This means that for $F_{a}, F_{b} \in Sp(2n,2)$ such that $a, b \in S_{3}\wr_{n}S_{n}$,  \begin{align}\label{eqn Ftau inverse}
            F_{a^{-1}}=F_{a}^{-1},
        \end{align} and
 \begin{align}\label{eqn alpha(ab) = FbFa}
            \alpha(abc) = \alpha(c)\alpha(b)\alpha(a) = F_{c}F_{b}F_{a}.
        \end{align}

Let $\tau \in \mathcal{T}$ such that $\tau: \mathcal{C} \rightarrow \mathcal{C}'$, where $\mathcal{C} \sim \mathcal{C}'$ i.e. $\mathcal{C}'$ is a distinct, equivalent ``version'' of $\mathcal{C}$. $\tau$ acts on the logical basis $\mathcal{B}$ as $\tau(\mathcal{B})$ and on $\pi \in \Gamma(\mathcal{C})$ via $\tau^{-1}\pi\tau$. Combining this action with Proposition \ref{eqn logical action}, we obtain Eq.~\eqref{eqn tau transform}: $L(\tau, \pi, \mathcal{B}) = \mathcal{D}_{\tau (\mathcal{B})} \odot ((\tau^{-1}\pi \tau)(\tau(\mathcal{B})))^{T}$, in which $L(\tau, \pi, \mathcal{B})$ denotes the logical action of the element $\pi' \in \Gamma(\mathcal{C}')$ on the encoded space of $\mathcal{C}'$ (where $\pi' = \tau^{-1}\pi\tau$ and $\mathcal{C}' = \tau(\mathcal{C})$).

Therefore, we acquire the following derivation:

$$L(\tau, \pi, \mathcal{B}) = \mathcal{D}_{\tau(\mathcal{B})}  \odot ((\tau^{-1}\pi\tau) (\tau(\mathcal{B})))^{T} $$

$$\underbrace{=}_{Fact \text{ } a.} \underbrace{\underbrace{\Omega_{2k}B}_{Eq.~\eqref{eqn destab}}F_{\tau}^{T}}_{Eq.~\eqref{eqn alpha:m pi to Fpi}} \Omega_{2n} (\underbrace{(\tau^{-1}\pi\tau)}_{Eq.~\eqref{eqn pi' = tauinv pi tau}}(\underbrace{BF_{\tau}^{T}}_{Eq.~\eqref{eqn alpha:m pi to Fpi}}))^{T}$$

$$= \underbrace{D_{B}}_{Eq.~\eqref{eqn destab}}F_{\tau}^{T} \Omega_{2n} (\underbrace{BF_{\tau}^{T}F_{\tau^{-1} \pi \tau}^{T}}_{Eq.~\eqref{eqn alpha:m pi to Fpi}})^{T}$$
$$= D_{B} F_{\tau}^{T} \Omega_{2n} F_{\tau^{-1}\pi\tau} F_{\tau} B^{T}$$
$$= D_{B}F_{\tau}^{T} \Omega_{2n} \underbrace{F_{\tau}F_{\pi}F_{\tau^{-1}}}_{Eq.~\eqref{eqn alpha(ab) = FbFa}}F_{\tau}B^{T}$$
$$=D_{B} (F_{\tau}^{T} \Omega_{2n} F_{\tau}) F_{\pi} (\underbrace{F_{\tau}^{-1}}_{Eq.~\eqref{eqn Ftau inverse}}F_{\tau})B^{T}$$
$$=D_{B} \underbrace{\Omega_{2n}}_{Fact \text{ } 4.} F_{\pi}B^{T}$$
$$\underbrace{=}_{Fact \text{ } a.} \mathcal{D}_{\mathcal{B}} \odot \underbrace{\pi(\mathcal{B})^{T}}_{Eq.~\eqref{eqn alpha:m pi to Fpi}}$$

This derivation proves that $L(\tau, \pi, \mathcal{B}) = L(\pi, \mathcal{B})$ i.e. code equivalence leaves logical operations invariant. 
\end{proof}

\section{Constructing the transversal subgroup $\mathcal{T}$ via 
\textit{GAP} and \textit{Magma}}\label{appx: transversal GAP}

The steps for computing $\mathcal{T}$ using the computer algebra systems \textbf{Magma} and \textbf{GAP} are as follows:
\begin{enumerate}
    \item Find the generators of the Hamming group in disjoint cycle notation
    \item Find the generators of the automorphism group in disjoint cycle notation (via \textbf{Magma})
    \item Construct the full $Ham(\mathcal{C})$ and $\Gamma(\mathcal{C})$ from their generators (via \textbf{GAP})
    \item Construct $\mathcal{T}$ by taking representative $\tau$ from each coset of $Ham(\mathcal{C})/\Gamma(\mathcal{C})$ (via \textbf{GAP})
\end{enumerate} We will illustrate the above steps by using the example of the $[[4,2,2]]$ code.

\subsubsection{Step $1$: Find the generators of the Hamming group in disjoint cycle notation.}

For the sake of utility, we express the elements of $Ham(\mathcal{C})$ as elements of $S_{3n}=S_{12}$. We can break down the generators of the Hamming group into $3$ types : \begin{enumerate}[i)]
    \item The generators of $S_{n}$
    \begin{itemize}
        \item the generators of $S_{4}$ are: $$(4,5,6,1,2,3,7,8,9,10,11,12)$$ $$(10,11,12,1,2,3,4,5,6,7,8,9)$$
    \end{itemize}
    \item The generators of scaling elements
    \begin{itemize}
        \item the generators of scaling elements of $S_{12}$ are $$(\textcolor{red}{2,3,1},4,5,6,7,8,9,10,11,12)$$ $$(1,2,3,\textcolor{red}{5,6,4},7,8,9,10,11,12)$$ $$(1,2,3,4,5,6,\textcolor{red}{8,9,7},10,11,12)$$ $$(1,2,3,4,5,6,7,8,9,\textcolor{red}{11,12,10})$$
    \end{itemize}
    \item The generators of conjugation elements
    \begin{itemize}
        \item the generators of conjugation as elements of $S_{12}$ are $$(\textcolor{red}{1,3,2},4,5,6,7,8,9,10,11,12)$$ $$(1,2,3,\textcolor{red}{4,6,5},7,8,9,10,11,12)$$ $$(1,2,3,4,5,6,\textcolor{red}{7,9,8},10,11,12)$$ $$(1,2,3,4,5,6,7,8,9,\textcolor{red}{10,12,11})$$
    \end{itemize}
\end{enumerate}

Combining the above three Hamming generator-types, we find that there are $10$ generators of $Ham(\mathcal{C})$,

 $$Ham(\mathcal{C}) \eqsim \text{ } < \begin{bmatrix}
     (4,5,6,1,2,3,7,8,9,10,11,12) \\
 (10,11,12,1,2,3,4,5,6,7,8,9) \\ 
 (\textcolor{red}{2,3,1},4,5,6,7,8,9,10,11,12) \\
 (1,2,3,\textcolor{red}{5,6,4},7,8,9,10,11,12) \\
 (1,2,3,4,5,6,\textcolor{red}{8,9,7},10,11,12) \\(1,2,3,4,5,6,7,8,9,\textcolor{red}{11,12,10}) \\
 (\textcolor{red}{1,3,2},4,5,6,7,8,9,10,11,12) \\
 (1,2,3,\textcolor{red}{4,6,5},7,8,9,10,11,12) \\(1,2,3,4,5,6,\textcolor{red}{7,9,8},10,11,12) \\
 (1,2,3,4,5,6,7,8,9,\textcolor{red}{10,12,11})\end{bmatrix} >.$$

 We express these generators in disjoint cycle notation, as this format is compatible with \textbf{GAP} input,

 $$Ham(\mathcal{C}) \eqsim \text{} < \begin{bmatrix}
     (1,4)(2,5)(3,6) \\
     (1,10,7,4)(2,11,8,5)(3,12,9,6) \\
     (1,2,3) \\
     (4,5,6) \\
     (7,8,9) \\
     (10,11,12) \\
     (2,3) \\
     (5,6) \\
     (8,9) \\
     (11,12)
 \end{bmatrix}  >.$$

\subsubsection{Step $2$:  Find the generators of the automorphism group in disjoint cycle notation (via \textbf{Magma})}

 We can obtain the generators of the automorphism group $\Gamma(\mathcal{C})$ via Magma. In order to construct $\Gamma(\mathcal{C})$, \textbf{Magma} requires us to input the generator $\mathcal{G}$ of the desired code $\mathcal{C}$. In the case of the $[[4,2,2]]$ code, we can use the following generator as input: $$\mathcal{G} = \begin{bmatrix}
     1 & 1 & 1 & 1 \\
     \omega & \omega & \omega & \omega 
 \end{bmatrix}.$$ 

 Below is a verbatim \textbf{Magma} command, specifying the $[[4,2,2]]$ code and obtaining its associated automorphism group $\Gamma(C)$. 
\begin{center}
\begin{verbatim}
    F<w>:=GF(4); 
    G:=Matrix(F,2,4,[1,1,1,1,w,w,w,w]);
    C:=AdditiveCode(GF(2),G); 
    Q:=QuantumCode(C); 
    AutomorphismGroup(Q);
\end{verbatim}
\end{center}

From this command, \textbf{Magma} should output 

\begin{center}
    \begin{verbatim}
    Permutation group acting on a set of cardinality 12 
    Order = 144 = 2^4 * 3^2 
    (7, 10)(8, 11)(9, 12) 
    (4, 10)(5, 11)(6, 12) 
    (2, 3)(5, 6)(8, 9)(11, 12) 
    (1, 3)(4, 6)(7, 9)(10, 12)  
    (1, 5, 3, 4, 2, 6)(7, 8, 9) 
    (10, 11, 12)
   \end{verbatim}
\end{center}

From the above output, we find the generators of the automorphism group $\Gamma(\mathcal{C})$ in disjoint cycle notation. 

$$\Gamma(\mathcal{C})= < \begin{bmatrix}
    (7, 10)(8, 11)(9, 12) \\
     (4, 10)(5, 11)(6, 12) \\
       (2, 3)(5, 6)(8, 9)(11, 12) \\
       (1, 3)(4, 6)(7, 9)(10, 12) \\
       (1, 5, 3, 4, 2, 6)(7, 8, 9) \\
       (10, 11, 12)
\end{bmatrix} >$$

Alternatively, we could also make use of Sayginel's autqec python package \cite{autqec_code} which maps the stabiliser generators to a related binary linear code and uses MAGMA or Bliss software to compute its automorphism group, which corresponds to the the physical implementation of logical Cliffords. This software also computes appropriate Pauli corrections to the physical circuits.

\subsubsection{Step $3$: Construct the full $Ham(\mathcal{C})$ and $\Gamma(\mathcal{C})$ from their generators (via \textbf{GAP})}
Now, we instruct \textbf{GAP} to construct the full Hamming group $Ham(\mathcal{C})$ from its generators in disjoint cycle notation.

\begin{verbatim}
    G:=Group(   (1,4)(2,5)(3,6),
    (1,10,7,4)(2,11,8,5)(3,12,9,6),
    (1,2,3),
    (4,5,6),
    (7,8,9),
    (10,11,12),
    (2,3),
    (5,6),
    (8,9),
    (11,12));
\end{verbatim}

We do the same for the automorphism group $\Gamma(C)$.

\begin{verbatim}
    A:= (   (7, 10)(8, 11)(9, 12), 
    (4, 10)(5, 11)(6, 12), 
    (2, 3)(5, 6)(8, 9)(11, 12), 
    (1, 3)(4, 6)(7, 9)(10, 12),  
    (1, 5, 3, 4, 2, 6)(7, 8, 9), 
    (10, 11, 12));
\end{verbatim}

\subsubsection{Step $4$: Construct $\mathcal{T}$ by taking representative $\tau$ from each coset of $Ham(\mathcal{C})/\Gamma(\mathcal{C})$ (via \textbf{GAP}) }

We can now use \textbf{GAP} to construct the transversal subgroup $\mathcal{T}$ by taking a representative $\tau$ of each coset $Ham(\mathcal{C}) / \Gamma(\mathcal{C})$. This is achieved by the following input command into \textbf{GAP}:

\begin{verbatim}
    t:=RightTransversal(G,A);
    List(t);
\end{verbatim}

\section{Construction of $F_{\pi}$ from $\pi \in S_{3} \wr_{n} S_{n}$}\label{appx symplectic construction}

The binary symplectic construction is similar to the formalism introduced by Dehaene and De Moor \cite{Dehaene}. Let the condensed symplectic form equivalent to $\pi$ be denoted as $\pi \eqsim (\sigma;\rho_{1}, \dots, \rho_{n}) \in (S_{3})^{n} \rtimes S_{n}$, where $\rho_{i} \in S_{3}$ corresponds to a local (single-qubit) Clifford action and $\sigma$ is the global permutation action on the $n$ blocks of $S_{3}$. From Eq.~\eqref{eqn alpha:m pi to Fpi}, we know that $\alpha$ is a homomorphism, where $\alpha(\pi) = F_{\pi}$. This homomorphism consists of two mappings $\beta$ and $\sigma$; $\beta$ is the local mapping acting on each $\rho_{i}$ and $\sigma$ is the global permutation acting on the positions of the $n$ blocks.

Each local $S_{3}$ factor $\rho_{i}$ is isomorphic to the single-qubit Clifford group modulo phase, which we can express as elements of $Sp(2,2)$. For a block $i \in \{1, \dots, n\}$, consisting of $3$ co-ordinates, the local map $\beta : S_3 \to Sp(2,2)$
corresponds to the action of each of the $6$ single-qubit Clifford operators on the Pauli basis \cite{Gottesman1998,Dehaene,AaronsonGottesman}. \begin{align}\label{eqn beta} \beta(\rho_{i})= \begin{bmatrix}
        a_{i} & b_{i} \\
    c_{i} & d_{i} 
    \end{bmatrix} \in Sp(2,2).\end{align} There are $|S_{3}|=6$ possible arrangements for each of the blocks $\rho_{i}$, each corresponding to a unique element $\beta(\rho_{i}) \in Sp(2,2)$. These $6$ possible arrangements and their corresponding $2 \times 2$ symplectic matrix be read off the table below:

   \begin{center}
   \tiny
   
\begin{tabular}{|p{1in}|p{1in}|p{1in}|p{1in}|p{1in}|}\hline \label{table internal permutations}

  $\rho_{i}$ &  $F(4)^{*}$ multiplication & Conjugation & $\alpha_{1}(\rho_{i}) \in Sp(2,2)$ & Corresponding Clifford\\

 \hline
 
 $1,2,0$ & $1$  & No & $\begin{bmatrix}
     1 & 0 \\
     0 & 1
 \end{bmatrix}$ & $I$ \\ 
 \hline
 $1,0,2$ & $1$  & Yes & $\begin{bmatrix}
     1 & 1 \\
     0 & 1
 \end{bmatrix}$  & $HSH$\\ 
 \hline
 $2,0,1$ & $\omega$  & No & $\begin{bmatrix}
     0 & 1 \\
     1 & 1
 \end{bmatrix}$ & $SH$ \\ 
 \hline
 $0,2,1$ & $\omega$ & Yes &  $\begin{bmatrix}
     1 & 0 \\
     1 & 1
 \end{bmatrix}$ & $S$ \\
 \hline
 $0,1,2$ & $\overline{\omega}$ & No & $\begin{bmatrix}
     1 & 1 \\
     1 & 0
 \end{bmatrix}$ & $HS$ \\
 \hline
 $2,1,0$ & $\overline{\omega}$ & Yes & $\begin{bmatrix}
     0 & 1 \\
     1 & 0
 \end{bmatrix}$ & $H$\\
 \hline
\end{tabular}
\end{center}

We can construct the block matrix $M \in Sp(2n,2)$ from $\beta(\rho_{i})$.  Each $\beta(\rho_{i}) \in Sp(2,2)$ preserves the single-qubit symplectic form and as a result, the full expression $M$ preserves the global symplectic form $\Omega_{2n} = \begin{bmatrix}
        0 & \mathbb{I}_{n} \\
        \mathbb{I}_{n} & 0
    \end{bmatrix}$. We define this $M$ as:
    \begin{align}\label{eqn m} M = \left[ \begin{array}{c|c}
A & B \\
\hline 
C & D
\end{array} \right], \end{align}
such that each $A,B,C,D \in \mathbb{F}_{2}^{2n \times 2n}$ and are constructed via $H_{ii}=1 \iff h_{i}=1$ for each $H \in \{A,B,C,D\}$ and their corresponding $h_{i} \in \{a_{i},b_{i}, c_{i}, d_{i}\}$. For example $A_{ii}=1 \iff a_{i} = 1$. In this form $M$ can represent all local Clifford actions and acts on $(x|z)^{T} \in \mathbb{F}_{2}^{2n}$ to produce $(x'|z')^{T}$ according to the following mechanism:

 \begin{align}
    M(x|z)^{T} = \left[ \begin{array}{ccc|ccc}
    a_{11} & \dots & a_{1n}  & b_{11} & \dots & b_{1n} \\
    \dots & \dots & \dots & \dots & \dots & \dots \\
    a_{n1} & \dots & a_{nn} & b_{n1} & \dots & b_{nn} \\
    \hline
    c_{11} & \dots & c_{1n} & d_{11} & \dots & d_{1n} \\
     \dots & \dots & \dots & \dots & \dots & \dots \\
    c_{n1} & \dots & c_{nn} & d_{n1} & \dots & d_{nn}
    \end{array} \right] \begin{bmatrix}
        x_{1} \\
        \dots \\
        x_{n} \\
        \hline 
        z_{1} \\
        \dots \\
        z_{n}
    \end{bmatrix} = \begin{bmatrix}
        a_{11}x_{1} + \dots + b_{1n}z_{n} \\
        \dots \\
        a_{n1}x_{1} + \dots + b_{nn}z_{n} \\
        \hline 
        c_{11}x_{1} + \dots + d_{1n}z_{n} \\
        \dots \\
        c_{n1}x_{1} + \dots + d_{nn}z_{n} \\
    \end{bmatrix} = (x'|z')^{T},
\end{align}  where \begin{align}
    x' = A.x + B.z
\end{align} \begin{align}
    z' = C.x + D.z.
\end{align}

The permutation $\sigma \in S_{n}$ acts on the $n$ blocks and treats each block as a single co-ordinate. The permutation matrix corresponding to $\sigma$ is $P_{\sigma} \in \mathbb{F}_{2}^{n}$. We can also express this permutation in terms of a block permutation matrix: $N \in Sp(2n,2)$
\begin{align}\label{eqn n}N= \left[ \begin{array}{c|c}
    P_{\sigma^{-1}} & 0 \\
    \hline
    0 & P_{\sigma^{-1}}
\end{array} \right]. \end{align}

The full symplectic representation of $\pi$ can be calculated as $F_{\pi}=NM$.  Both $N$ and $M$ preserve the symplectic form
$\Omega_{2n} $ 
\cite{Gottesman1997, Dehaene}, 
and therefore their product $F_\pi$ also satisfies:
\begin{align}
    F_\pi \, \Omega_{2n} \, F_\pi^{T} = \Omega_{2n},
\end{align}
which ensures that $F_\pi \in Sp(2n,2)$.

\subsection{Example}\label{subsec example converting pi to fpi}

In section \ref{section example with fig}, we give the example of the $[[4,2,2]]$ code. We found that \begin{center} $\pi^{46} \eqsim$
\begin{quantikz}[background color=black!5!white]
\lstick{}  & \push{  \begin{bsmallmatrix}
    1 & 1 \\
    0 & 1
\end{bsmallmatrix}} & \permute{2,3,4,1} & \push{} \\
\lstick{} &  \push{  \begin{bsmallmatrix}
    1 & 1 \\
    0 & 1
\end{bsmallmatrix}} & \push{} & \push{} \\
\lstick{} &  \push{  \begin{bsmallmatrix}
    1 & 1 \\
    0 & 1
\end{bsmallmatrix}}  & \push{} & \push{} \\
\lstick{} &   \push{  \begin{bsmallmatrix}
    1 & 1\\
    0 & 1
\end{bsmallmatrix}} & \push{} & \push{}.
\end{quantikz}
\end{center} 

We can express $\pi^{46}$ as an element of $ \in S_{3} \wr_{n} S_{n}$, $$\pi^{46} = (4,6,5,7,9,8,10,12,11,1,3,2).$$ 
The automorphism is arranged in $n$ blocks of $3$, as described in colour below: 
$$\pi^{46} = (\textcolor{red}{4,6,5} \textcolor{blue}{,7,9,8} \textcolor{green}{,10,12,11}\textcolor{magenta}{,1,3,2}).$$
We refer to the ordering of the $3$ elements within each block as the \textit{internal permutation} and we refer to the overall ordering of the $n$ blocks as the \textit{global permutation}. We begin by describing the global permutation, demonstrated by the labels shown below:
   $$\pi^{46} = (\underset{2}{\textcolor{red}{4,6,5}} \underset{3}{\textcolor{blue}{,7,9,8}} \underset{4}{\textcolor{green}{,10,12,11}}\underset{1}{\textcolor{magenta}{,1,3,2}}).$$
Clearly the permutation part of $\pi^{46}$ is $\sigma= (2,3,4,1)$. We construct the inverse of this permutation, $\sigma^{-1} = (4,1,2,3)$, meaning that $$P_{\sigma^{-1}} = \begin{bmatrix}
       0 & 0 & 0 & \textcolor{magenta}{1} \\
       \textcolor{red}{1} & 0 & 0 & 0 \\
       0 & \textcolor{blue}{1} & 0 & 0 \\
       0 & 0 & \textcolor{green}{1} & 0 
   \end{bmatrix}.$$ We use $P_{\sigma}^{-1}$ to construct the global permutation as described in Eq.~\eqref{eqn n}
   $$N = \left[ \begin{array}{cccc|cccc}
     0 & 0 & 0 & \textcolor{magenta}{1} & 0 & 0 & 0 & 0 \\
     \textcolor{red}{1} & 0 & 0 & 0 & 0 & 0 & 0 & 0 \\
     0 & \textcolor{blue}{1} & 0 & 0 & 0 & 0 & 0 & 0 \\
     0 & 0 & \textcolor{green}{1} & 0 & 0 & 0 & 0 & 0 \\ \hline
     0 & 0 & 0 & 0 & 0 & 0 & 0 & \textcolor{magenta}{1} \\
     0 & 0 & 0 & 0 & \textcolor{red}{1} & 0 & 0 & 0 \\
     0 & 0 & 0 & 0 & 0 & \textcolor{blue}{1} & 0 & 0 \\
     0 & 0 & 0 & 0 & 0 & 0 & \textcolor{green}{1} & 0
   \end{array}\right].$$
We now deetermine the internal permutation of each block. We can rewrite the values $\mod 3$ to obtain the $n=4$ blocks $\rho_{i} \in S_{3}$. We can see that the first block is $(4,6,5) \equiv (1,0,2) \mod 3 \implies \rho_{1} = (1,0,2)$ and as such, upon consultation with table \ref{table internal permutations}, we observe that $\beta(\rho_{1})=\begin{bmatrix}
    1 & 1 \\
    0 & 1
    \end{bmatrix}$ which corresponds to a physical $HSH$ gate acting on qubit $1$.

In this manner and according to Eq.~\eqref{eqn beta}, construct the single qubit symplectic matrices corresponding to each $\beta(\rho_{i})$: 
$$\beta(\rho_{1}) = \textcolor{red}{\begin{bmatrix}
    1 & 1 \\
    0 & 1
\end{bmatrix}}$$

$$\beta(\rho_{2})= \textcolor{blue}{\begin{bmatrix}
    1 & 1 \\
    0 & 1
\end{bmatrix}}$$

$$\beta(\rho_{3})=\textcolor{green}{\begin{bmatrix}
    1 & 1 \\
    0 & 1 
\end{bmatrix}}$$

$$\beta(\rho_{4})=\textcolor{magenta}{\begin{bmatrix}
    1 & 1 \\
    0 & 1
\end{bmatrix}}$$ We can now construct $M \in Sp(2n,2)$ according to Eq.~\eqref{eqn m} $$M = \left[\begin{array}{cccc|cccc}
    \textcolor{red}{1} & 0 & 0 & 0 & \textcolor{red}{1} & 0 & 0 & 0 \\
    0 & \textcolor{blue}{1} & 0 & 0 & 0 & \textcolor{blue}{1} & 0 & 0 \\
    0 & 0 & \textcolor{green}{1} & 0 & 0 & 0 & \textcolor{green}{1} & 0 \\
    0 & 0 & 0 & 1 & 0 & 0 & 0 & \textcolor{magenta}{1} \\ \hline
    0 & 0 & 0 & 0 & \textcolor{red}{1} & 0 & 0 & 0 \\
    0 & 0 & 0 & 0 & 0 & \textcolor{blue}{1} & 0 & 0 \\
    0 & 0 & 0 & 0 & 0 & 0 & \textcolor{green}{1} & 0 \\
    0 & 0 & 0 & 0 & 0 & 0 & 0 & \textcolor{magenta}{1} 
\end{array} \right].$$

We use $N$ and $M$ to construct the symplectic matrix $$F_{\pi^{46}} = NM = \begin{bmatrix}
    0 & 0 & 0 & 1 & 0 & 0 & 0 & 1 \\
    1 & 0 & 0 & 0 & 1 & 0 & 0 & 0 \\
    0 & 1 & 0 & 0 & 0 & 1 & 0 & 0 \\
    0 & 0 & 1 & 0 & 0 & 0 & 1 & 0 \\
    0 & 0 & 0 & 0 & 0 & 0 & 0 & 1 \\
    0 & 0 & 0 & 0 & 1 & 0 & 0 & 0 \\
    0 & 0 & 0 & 0 & 0 & 1 & 0 & 0 \\
    0 & 0 & 0 & 0 & 0 & 0 & 1 & 0 \\
\end{bmatrix}.$$

We demonstrate how $F_{\pi^{46}}$ can be used to derive the logical effect of $\pi^{46}$ acting on the basis $\mathcal{B}$, given in the example in section \ref{section example with fig}. $F_{\pi^{46}}$ acts on the basis $\mathcal{B} = \begin{bmatrix}
    1 & 1 & 0 & 0 \\
    \omega & \omega & 0 & 0 \\
    \omega & 0 & \omega & 0 \\
    1 & 0 & 1 & 0
\end{bmatrix}$ by the following action: $$BF_{\pi^{46}}^{T} = \begin{bmatrix}
    1 & 1 & 0 & 0 & 0 & 0 & 0 & 0 \\
    0 & 0 & 0 & 0 & 1 & 1 & 0 & 0 \\
    0 & 0 & 0 & 0 & 1 & 0 & 1 & 0 \\
    1 & 0 & 1 & 0 & 0 & 0 & 0 & 0
\end{bmatrix} \begin{bmatrix}
    0 & 1 & 0 & 0 & 0 & 0 & 0 & 0 \\
    0 & 0 & 1 & 0 & 0 & 0 & 0 & 0 \\
    0 & 0 & 0 & 1 & 0 & 0 & 0 & 0 \\
    1 & 0 & 0 & 0 & 0 & 0 & 0 & 0 \\
    0 & 1 & 0 & 0 & 0 & 1 & 0 & 0 \\
    0 & 0 & 1 & 0 & 0 & 0 & 1 & 0 \\
    0 & 0 & 0 & 1 & 0 & 0 & 0 & 1 \\
    1 & 0 & 0 & 0 & 1 & 0 & 0 & 0 
\end{bmatrix} = \begin{bmatrix}
    0 & 1 & 1 & 0 & 0 & 0 & 0 & 0 \\
    0 & 1 & 1 & 0 & 0 & 1 & 1 & 0 \\
    0 & 1 & 0 & 1 & 0 & 1 & 0 & 1 \\
    0 & 1 & 0 & 1 & 0 & 0 & 0 & 0
\end{bmatrix}.$$

We have shown that $$\pi^{46}(\mathcal{B})= \begin{bmatrix}
    0 & 1 & 1 & 0 \\
    0 & \omega^{2} & \omega^{2} & 0 \\
    0 & \omega^{2} & 0 & \omega^{2} \\
    0 & 1 & 0 & 1
\end{bmatrix},$$ which confirms the logical operation $L(\pi^{46},\mathcal{B})$ as
 $\mathcal{B}_{1}=(0,1,1,0) \in [XZ]$, $\mathcal{B}_{2}=(0, \omega^{2}, \omega^{2}, 0) \in [YY]$, $\mathcal{B}_{3}=(0, \omega^{2}, 0, \omega^{2}) \in [ZZ]$ and $\mathcal{B}_{4}=(0,1,0,1) \in [IZ]$, which confirms that $$L(\pi^{46},\mathcal{B}) = \begin{bmatrix}
    1 & 1 & 0 & 0 \\
    0 & 1 & 0 & 0 \\
    0 & 1 & 1 & 0 \\
    1 & 1 & 1 & 1
\end{bmatrix}.$$

\end{document}